\documentclass[11pt]{article}

\usepackage{times,amsthm}
\usepackage{amssymb,verbatim,ifthen}
\usepackage[all]{xy}
\usepackage{graphicx}
\usepackage{amsmath}
\usepackage{amssymb}
\usepackage{enumerate}
\usepackage{mdwlist}
\usepackage{commath}
\usepackage{mathrsfs}
\usepackage[ruled,vlined]{algorithm2e}

\def\Split{true}

\ifthenelse{\equal{\Split}{true}} {
\textwidth = 450pt \oddsidemargin = 2pt \evensidemargin = 2pt

} {
\usepackage{fullpage}

}

\newtheorem{theorem}{Theorem}
\newtheorem{lemma}[theorem]{Lemma}

\newtheorem{proposition}[theorem]{Proposition}
\newtheorem{claim}[theorem]{Claim}

\newtheorem{remark}[theorem]{Remark}
\newtheorem{example}[theorem]{Example}

\theoremstyle{definition}
\newtheorem{definition}[theorem]{Definition}

\newcommand{\F}{{\mathbb{F}}}
\newcommand{\N}{\mathbb{N}}
\newcommand{\Z}{\mathbb{Z}}
\newcommand{\R}{\mathbb{R}}
\newcommand{\C}{\mathbb{C}}

\DeclareMathOperator*{\E}{\mathbb{E}}
\newcommand{\half}{\text{\sf half}}

\newcommand{\LSB}{\text{LSB}}
\newcommand{\MSB}{\text{MSB}}

\newcommand{\famly}{\mathscr{J}}

\newcommand{\noit}{\operatorname}

\def\wp{\omega_p}

\newcommand{\bit}{\text{bit}}

\makeatletter
\def\moverlay{\mathpalette\mov@rlay}
\def\mov@rlay#1#2{\leavevmode\vtop{%
   \baselineskip\z@skip \lineskiplimit-\maxdimen
   \ialign{\hfil$\m@th#1##$\hfil\cr#2\crcr}}}
\newcommand{\charfusion}[3][\mathord]{
    #1{\ifx#1\mathop\vphantom{#2}\fi
        \mathpalette\mov@rlay{#2\cr#3}
      }
    \ifx#1\mathop\expandafter\displaylimits\fi}
\makeatother

\newcommand{\cupdot}{\charfusion[\mathbin]{\cup}{\cdot}}

\usepackage{titlesec}
\titleformat{\section}{\large\scshape\centering}{\thesection.}{1em}{}
\titleformat{\subsection}{\normalsize\bfseries}{\thesubsection}{1em}{}

\titleformat{\paragraph}
{\normalfont\normalsize\bfseries}{\theparagraph}{1em}{}
\titlespacing*{\paragraph}
{0pt}{3.25ex plus 1ex minus .2ex}{1.5ex plus .2ex}

\def\_{{\tt \char`\_ }}

\usepackage{tikz}

\begin{document}
\pagestyle{plain} % \lineskip=1pt\baselineskip=15pt\lineskiplimit=2pt
\lineskip=1.8pt\baselineskip=18pt\lineskiplimit=0pt %\count0=1

\title{Finding Significant Fourier Coefficients: Clarifications, Simplifications, Applications and Limitations}

\author{Steven D. Galbraith, Joel Laity and Barak Shani \\ \small{Department of Mathematics, University of Auckland, New Zealand}}
%\address{Department of Mathematics, University of Auckland, New Zealand}
\date{}
\maketitle

%\bigskip\bigskip\bigskip
\begin{abstract}
    Ideas from Fourier analysis have been used in cryptography for the last three decades. Akavia, Goldwasser and Safra unified some of these ideas to give a complete algorithm that finds significant Fourier coefficients of functions on any finite abelian group. Their algorithm stimulated a lot of interest in the cryptography community, especially in the context of ``bit security''. This manuscript attempts to be a friendly and comprehensive guide to the tools and results in this field.
    The intended readership is cryptographers who have heard about these tools and seek an understanding of their mechanics and their usefulness and limitations.
    A compact overview of the algorithm is presented with emphasis on the ideas behind it. We show how these ideas can be extended to a ``modulus-switching'' variant of the algorithm. We survey some applications of this algorithm, and explain that several results should be taken in the right context. In particular, we point out that some of the most important bit security problems are still open. Our original contributions include: a discussion of the limitations on the usefulness of these tools; an answer to an open question about the modular inversion hidden number problem.

\end{abstract}
%\bigskip

\newpage
\tableofcontents
\lineskip=1.8pt\baselineskip=18pt\lineskiplimit=0pt %\count0=1

\section{Introduction}

Let $G$ be a finite abelian group.
Fourier analysis provides a convenient basis for the space of functions $G \to \C$,  namely the characters $\chi: G \to \C$. It follows that any function $f: G \to \C$ can be represented as a linear combination $f(x) = \sum_{\alpha \in G} \widehat{f}(\alpha) \chi_\alpha(x)$, where $\widehat{f}$ is the discrete Fourier transform of $f$. A standard problem is to approximate a function, up to any error term, using a linear combination of a small number of characters. This is not always possible, but for certain functions (which are called \emph{concentrated}) it is possible. The coefficients in such an approximation are called \emph{significant} Fourier coefficients, as their size is large relative to the function's norm.
The simplest example of a concentrated function is a character itself.

A natural computational problem is to compute such an approximation.
When doing this one might have a complete description of the function or, as will be the case in this paper, just a small set of values $f(x_i)$. The ability to choose specific $x_i$'s plays a crucial role in the ability to approximate $f$. Indeed, the main result in this subject is an algorithm that, given the ability to select the values $x_i$, efficiently computes a \emph{sparse} approximation for any concentrated function on any abelian group $G$, by computing all its significant coefficients. On the other hand, when the $x_i$'s cannot be selected, such an algorithm is not known to exist in general. Furthermore it is conjectured that an efficient algorithm does not exist in the general case.

We use the general term \emph{significant Fourier transform} (SFT) to refer to algorithms that compute a function's significant coefficients. SFT algorithms first appear explicitly in the work of %Goldreich--Levin~\cite{GL} and
Kushilevitz and Mansour~\cite{KM}, though some of the main ideas already appear in earlier works.
Subsequently new algorithms were presented, in various special cases of groups or functions, until the work of Akavia, Goldwasser and Safra~\cite{AGS} who presented a generic algorithm for all finite abeliean groups and all complex-valued functions.
The algorithms in the literature are often presented very differently, and some of them are designed to fulfill a very particular task, but they are all based on the same mathematical principles.

The main aim of this paper is to present a complete study of the SFT algorithms. Our work unifies these algorithms by clarifying the core mathematics underlying them. Thus, our focus is on a broad mathematical overview using Fourier analysis on finite groups and elementary group theory. We remark that our work is not necessarily the best presentation of a specific SFT algorithm, but we believe that a reader who is interested in understanding the rules and framework of these algorithms would benefit from this work. Our study also leads to a new approach for some of the more complicated cases. Furthermore, this paper surveys applications of the SFT algorithm in the field of cryptography and also gives limitations for such applications.

The SFT algorithm and variants have received great attention in the literature outside the regime of cryptography.
% The Kushilevitz--Mansour algorithm is a cornerstone in this research field, and serves as a basis for most algorithms, including the one we present in this paper.
Researchers in engineering, concerned with practical applications in signal processing, have developed algorithms with greater efficiency (with respect to various metrics); for a recent survey on these algorithms see Gilbert, Indyk, Iwen and Schmidt~\cite{Fourier_survey}. Our work does not cover these developments.

\vskip 0.2cm

\noindent \textbf{Roadmap}

Section~\ref{sec:basics} summarises the basic definitions. Section~\ref{sec_algorithm} presents the key ideas behind the SFT algorithm, and deals with some related issues. Specifically, with few a examples we explain why being able to choose the inputs to the functions is essential and why one does not expect to have a similar tool when the inputs to the functions are chosen at random; In cases where the function values are given by an oracle, we analyze the case of working with unreliable oracles.
% This section starts by presenting some of the different SFT algorithms and highlights their connections.

Section~\ref{sec:history} reviews the development of ideas and highlights the contributions of
Goldreich and Levin~\cite{GL}, Kushilevitz and Mansour~\cite{KM}, Mansour~\cite{Mansour92}, Bleichenbacher~\cite{Bleichenbacher} and Akavia, Goldwasser and Safra~\cite{AGS}.

In Section~\ref{sec:mod-switch} we outline our recent work~\cite{LS} on applying \emph{modulus switching} to this subject (namely to re-cast a function on $\Z_p$ to a function on $\Z_{2^n}$ for the nearest power of $2$ to $p$).
These ideas are very similar to the approach taken in Shor's (period-finding) algorithm~\cite{Shor}.
The benefit of this new approach is twofold. Firstly, its analysis gives insights into the AGS algorithm. Secondly it provides a new approach for implementations and for proving concentration of functions. In particular we provide a new proof of a result by Morillo and R{\` a}fols~\cite{MR} (described in Section~\ref{sec:i-th-bit}). %in our opinion this provides a simpler approach to the results of Akavia--Goldwasser--Safra and Morillo--R{\` a}fols.

%\noindent \textbf{Applications.}
The SFT algorithm is a useful tool in the research area of bit security.
%, as already mentioned above.
Section~\ref{sec_applications} surveys bit security applications using the language of the \emph{hidden number problem}: given $f$ and oracle access to $f_s := f \circ \varphi_s$, for some function $\varphi$ parameterized by an unknown value $s$, recover the value $s$.
The main application is in the group $G=\Z_p$ for the particular function $\varphi_s(x) = sx \pmod{p}$, i.e. $f_s := f(s x)$. In this particular case the \emph{scaling property} gives $\widehat{f_s}(\alpha) = \widehat{f}(\alpha s^{-1})$ for every $\alpha \in G$. It follows that $f$ and $f_s$ share the same coefficients in different order. If $\alpha$ is a significant Fourier coefficient of $f$ and $\beta$ is a significant Fourier coefficient of $f_s$ then $\alpha \beta^{-1}$ is a candidate value for $s$.

Using this observation, Akavia, Goldwasser and Safra~\cite{AGS} showed that a number of bit security results (for RSA, Rabin, and discrete logs) can be re-proved using these tools.
A classic result of this type, from Alexi, Chor, Goldreich and Schnorr (ACGS) \cite{ACGS}, is that if one has an oracle that on input $x^e \pmod{N}$ (where $(N,e)$ is an RSA public key) returns the least significant bit of $x$ with probability noticeably better than $\tfrac{1}{2}$, then one can compute $e$-th roots modulo $N$.
H\r{a}stad and N{\"a}slund \cite{HN} generalized this result for an oracle that returns any single bit of $x$ (see also \cite[Section 4.1]{HardCoreSurvey}), but their method is very complex and requires complicated and adaptive manipulations of the bits.
%Most of the above work used a small pool of functions (usually $(-1)$ raised to the power of a noisy inner product modulo 2 or else the most or least significant bit MSB/LSB); see Example~\ref{example4}.
%It is easy to show that if $\bit_i : \Z_{2^n} \to \Z_2$ denotes the $i$-th bit in the binary expansion of an integer, then $(-1)^{\bit_i(x)}$ is a concentrated function on $\Z_{2^n}$; see  Section \ref{sec:i-th-bit}.
On the other hand, the algorithm given by AGS, which applies to functions with significant Fourier coefficients, is much clearer and is not adaptive.\footnote{We describe the notion of adaptiveness in Section \ref{sec_applications}.} Similar to H\r{a}stad and N{\"a}slund, Morillo and R\`{a}fols~\cite{MR} extended the AGS results to all single bit functions, by showing that each single bit function is concentrated and so has a significant Fourier coefficient (in particular, one can obtain the ACGS result for any bit).
% dramatically extended the class of concentrated functions on $\Z_p$, by showing that the $i$-th bit function on $\Z_p$ is concentrated.
%Once this fact was known, the SFT algorithm allows to
The SFT algorithm has also been used to show \emph{search-to-decision} reductions for the \emph{learning with errors} and \emph{learning with rounding} problems~\cite{MM,LWR_reduction}.

Subsequently, a number of papers~\cite{Duc,Fazio,MVHNP,WZZ} have proved (or re-proved) various results on bit security in the context of Diffie--Hellman keys on elliptic curves and finite fields $\F_{p^n}$ with $n > 1$, but these results consider an unconventional model that allows changing the curve or field representation.
We emphasize that the requirement of chosen inputs for the functions restricts these applications. Indeed, the question of main interest, whether single bits of Diffie--Hellman shared keys are hardcore in a fixed representation, is still open. We elaborate on these applications in Section \ref{sec_applications}.

%\noindent \textbf{Limitations.}
Section~\ref{sec_limitations} explains a fundamental limitation to the approach described above: we prove that one can only solve the (chosen-multiplier) hidden number problem with these tools when the function $\varphi_s$ is linear or affine.
%) change of variable for the function $f$. That is, $\varphi_s(x) = sx + t$ in $\Z_N$ and $\varphi_s(x) = \langle s, x \rangle + t$ in $\Z^n_N$, for some constant value $t$.
Therefore, these tools cannot be directly used to address the \emph{elliptic curve hidden number problem} or the \emph{modular inversion hidden number problem}. Our work therefore answers a question in~\cite{MIHNP}.

\section{Preliminaries} \label{sec:basics}

The following gives mathematical background needed to understand the paper and definitions that will be used throughout the paper. The main definitions and notation appear in the table in Section~\ref{sec:table}.

\subsection{Fourier analysis on finite groups}

We review basic background on Fourier analysis on discrete domains. Proofs and further details can be found in Terras~\cite{Terras}.

Let $(R, +, \cdot)$ be a finite ring and denote by $G:=(R,+)$ the corresponding additive abelian group. We are interested in the set of functions $L^2(R) := \{ f:R \rightarrow \C \}$. The set $L^2(R)$ is a vector space over $\C$ of dimension $\abs{R}$, with the usual pointwise addition and scalar multiplication of functions.
Convolution of two functions $f, g\in L^2(R)$ is defined by $(f * g)(x) = \frac{1}{\abs{R}}\sum_{y\in R}f(x-y)g(y)$.
The \emph{expectation} of a function $f \in L^2(R)$ is defined to be $\E\left[f\right] = \frac{1}{\abs{R}} \sum_{x \in R} f(x)$.
The space $L^2(R)$ is equipped with an inner product $\langle f, g \rangle := \E\left[ f(x) \overline{g(x)} \right] = \frac{1}{\abs{R}} \sum_{x \in R} f(x) \overline{g(x)}$, where $\overline{z}$ denotes the complex conjugate of $z\in \C$. The inner product induces a norm $\|f\|_2 = \sqrt{\langle f, f\rangle}$.
We also define $\| f \|_\infty = \max_{x \in R} |f(x)|$.

One basis for this vector space is the set of \emph{Kronecker delta functions} $\{\delta_i\}_{i \in R}$ $\left(\delta_i(j) = 1\right.$ if $j=i$, otherwise $\left. \delta_i(j) = 0\right)$.
This is an orthogonal basis with respect to the inner product.
% And $\delta_i * \deta_j = \delta_{i+j} / |R|$.
% form a basis for this vector space; every function $f:R \rightarrow \C$ can be written as $f(x) = \sum_{i \in R} f(i) \delta_i(x)$.
However, this basis is not as useful as the Fourier basis, as we will explain later in this section.

A \emph{character} of an additive group $G$ is a group homomorphism taking values in the non-zero complex numbers, namely $\chi: G \rightarrow \C^*$ such that $\chi(x+y) = \chi(x) \chi(y)$. Since $\chi(x)^{\abs{G}} = \chi(\abs{G} x) = \chi(0_G) = 1$, we see that the characters take values in the complex $\abs{G}$-th roots of unity. The set of characters of $G$ forms a group (with respect to pointwise multiplication), isomorphic to $G$, which is often denoted $\widehat{G}$.

In general, we fix a choice of isomorphism $G \to \widehat{G}$ and denote it by $\alpha \mapsto \chi_\alpha$.
In particular, for $G=\Z_N$ the characters are defined by $\chi_\alpha (x) := \mathrm{e}^{\frac{2\pi i}{N} \alpha x}$ where $\alpha\in G$. For $G = \Z_{N_1} \times \ldots \times \Z_{N_m}$, let $\boldsymbol{\alpha}= \left(\alpha_1,\dots,\alpha_m \right)$ and $\boldsymbol{x} = (x_1,\dots,x_m)$; the character $\chi_{\boldsymbol{\alpha}}$  is given by $\chi_{\boldsymbol{\alpha}}(\boldsymbol{x}) :=  \chi_{\alpha_1}(x_1) \cdot \ldots \cdot \chi_{\alpha_m}(x_m) = \mathrm{e}^{\frac{2\pi i}{N_1} \alpha_1 x_1} \cdot \ldots \cdot \mathrm{e}^{\frac{2\pi i}{N_m} \alpha_m x_m}$ and the map $\boldsymbol{\alpha}\mapsto \chi_{\boldsymbol{\alpha}}$ from $G$ to $\widehat{G}$ is an isomorphism.
We sometimes write $\omega_N := \mathrm{e}^{\frac{2\pi i}{N}}$ so that $\chi_\alpha(x) = \omega_N^{\alpha x}$.
%We may associate each group element $\alpha \in G$ with a character $\chi_\alpha$.

%
%That is, denote by $\widehat{G}$ the set (group) of characters of $G$, and consider the map $\varphi: G \rightarrow \widehat{G}$, given by $\varphi(a) := \chi_a$. The map $\varphi$ can be shown to be an isomorphism.

The following relations are standard and can be used to show that the characters are orthonormal
\[
    \sum_{x\in G}\chi(x)=
    \begin{cases}
        |G| & \text{if }\chi \text{ is the identity in }\widehat{G}, \\
        0   & \text{otherwise},
    \end{cases}
    \quad \quad
    \sum_{\chi\in \widehat{G}}\chi(x)=
    \begin{cases}
        |G| & \text{if }x=0,    \\
        0   & \text{otherwise}.
    \end{cases}
\]
If $G =\Z_{N_1} \times \ldots \times \Z_{N_m}$ then for any subgroup $H \leq G$ we define the orthogonal set
\begin{equation}\label{eq_Hperp}
    H^\perp := \{a\in G\mid \chi_a(h) = 1 \text{ for all }h\in H\} \,.
\end{equation}
This  set is fundamental for the understanding of the SFT algorithm and appears frequently in Section \ref{sec:SFTalgo}. Using the relations above it can be shown that
\begin{equation}\label{eq:Hexp}
    \sum_{h\in H}\chi_h(x) = \begin{cases} |H|, & \text{if }x\in H^\perp, \\ 0, &\text{otherwise}.\end{cases}
\end{equation}

%(XXX HERE IS THE PROOF ADDED BY JOEL. Is it necessary?  XXXX)
%For completeness we sketch a proof of equation~(\ref{eq:Hexp}).
%   If $x\in H^\perp$ then, by definition of $H^\perp$ and the property $\chi_x(y) = \chi_y(x)$ we have $\chi_h(x)=1$ for all $h\in H$. If $x\not\in H^\perp$ then, there exists some $h'\in H$ for which $\chi_{x}(h') \neq 1$. Then
%   $
%   \sum_{h\in H}\chi_{h}(x) = \sum_{h\in H}\chi_{h+h'}(x) = \chi_{h'}(x)\sum_{h\in H}\chi_h(x)=\chi_x(h')\sum_{h\in H}\chi_h(x)
%   $
%   which implies that $\sum_{h\in H}\chi_{h}(x) =0$.

The \emph{Fourier basis} for $L^2(R)$ is the set $\widehat{G}$ consisting of all the characters $\chi$.
It is an orthonormal basis. Therefore, we can represent each function $f: R \rightarrow \C$ uniquely as a linear combination $f(x) = \sum_{\alpha \in G} \widehat{f}(\alpha) \chi_\alpha(x)$ of the characters $\chi_\alpha$.
The function $\widehat{f}:G\to \C$ given by $\widehat{f}( \alpha ) = \langle f,\chi_\alpha \rangle$ is called the \emph{discrete Fourier transform}.
The map $f \mapsto \widehat{f}(\alpha)$ is $\C$-linear.
%Let $\overline{\chi}_a$ be the conjugate to the character $\chi_a$. That is, $\overline{\chi}_a(x) = \overline{\chi_a(x)}$.
Notice that a single Fourier coefficient encapsulates information about the function on the whole domain, unlike the representation in terms of Kronecker delta functions where one coefficient only holds information about the function at a single point.

\emph{Parseval's identity} is the following relationship between the norms of $f$ and $\widehat{f}$:
\[
    \Vert f \Vert_2^2 =  \frac{1}{|G|} \sum_{x \in G} | f(x)|^2 = \langle f, f \rangle = \sum_{\alpha \in G} | \widehat{f}(\alpha) |^2 = |G| \cdot \|\widehat{f}\|_2^2 \,.
\]

Adopting signal-processing terminology, when we work with the values $f(x)$ for $x\in G$ we say that $x$ is in the \emph{time domain}. When we use the values $\widehat{f}(\alpha)$ we say $\alpha\in G$ is in the \emph{frequency domain}.
There does not seem to be a rigorous formulation of this terminology and we do not use it much, but the reader will find it very common in the engineering literature.
We signal to the reader whether we are working in the time domain or frequency domain by using Latin letters $x,y$ for elements in the former (elements of $G$), and Greek letters $\alpha,\beta$ for the latter (corresponding to elements of $\widehat{G}$, e.g. $\chi_{\alpha}$).

% to the frequency domain. We adopt this language and often call `time domain' to the original domain of $f$, and `frequency domain' to the space of transformed functions. In other words, $f$ belongs to the `time domain' while $\widehat{f}$ belongs to the `frequency domain'.

Let $R=\Z_{N_1} \times \ldots \times \Z_{N_m}$ with componentwise addition and multiplication, and let $f,g\in L^2(R)$.
Basic properties of the Fourier transform include the following (note that the basis of Kronecker delta functions does not satisfy these properties, which is one reason why it is less useful than the Fourier basis):

\begin{itemize}
    \item (time) scaling: if $g(x) := f(cx)$ for $c \in R^*$, then $\widehat{g}(\alpha) = \widehat{f}(c^{-1}\alpha)$;
    \item (time) shifting: if $g(x) := f(c+x)$ for $c \in R$, then $\widehat{g}(\alpha) = \widehat{f}(\alpha)\chi_\alpha(c)$;
    \item (frequency) shifting: if $g(x) := f(x)\chi_c(x)$ for $c \in R$, then $\widehat{g}(\alpha) = \widehat{f}(\alpha-c)$;
    \item convolution-multiplication duality: $\widehat{f \ast g}(\alpha) = \widehat{f}(\alpha) \widehat{g}(\alpha)$.
\end{itemize}

We now recall some definitions from \cite{AGS, Duc, MR}. The same definitions can be made for functions over rings $R$ where $G$ is their additive group.

\begin{definition} [Restriction]
    Given a function $f:G \rightarrow \C$ and a set of characters $\Gamma \subseteq \widehat{G}$, the \emph{restriction of $f$ to $\Gamma$} is the function $f|_\Gamma: G \rightarrow \C$ defined by $f|_\Gamma := \sum_{\chi_\alpha \in \Gamma}\widehat{f}(\alpha) \chi_\alpha$.
\end{definition}

%\begin{definition} [Concentration]
%Let $\famly = \{f_i:G_i\to \C\}_{i\in \N}$ be a family of functions. Let $\epsilon>0$ be a real number.
%Then $\famly$ is \emph{Fourier $\epsilon$-concentrated} if there exist a polynomial $P \in \Z[s,t]$ and sets of characters $\Gamma_i \subseteq \widehat{G}_i$ such that $|\Gamma_i| \le P\Bigl( \tfrac{1}{\epsilon} , \log|G_i| \Big)$ and $\| f_i - f_i|_{\Gamma_i} \|_{2}^{2} \leq \epsilon$ for all $i\in \N$. We say that $\famly$ is \emph{concentrated} if $\famly$ is $\epsilon$-concentrated for every $\epsilon > 0$.
%\end{definition}

\begin{definition} [$\epsilon$-Concentration]
    Let $\epsilon>0$ be a real number. A family of functions $\{f_i:G_i\to \C\}_{i\in \N}$ is \emph{Fourier $\epsilon$-concentrated} if there exists a polynomial $P$ and sets of characters $\Gamma_i \subseteq \widehat{G}_i$ such that $|\Gamma_i| \le P( \log|G_i| )$ and $\| f_i - f_i|_{\Gamma_i} \|_{2}^{2} \leq \epsilon$ for all $i\in \N$.
\end{definition}

\begin{definition} [Concentration]
    A family of functions $\{f_i:G_i\to \C\}_{i\in \N}$ is \emph{Fourier concentrated} if there exists a polynomial $P$ and sets of characters $\Gamma_i \subseteq \widehat{G}_i$ such that $|\Gamma_i| \le P( \log|G_i|/\epsilon )$ and $\| f_i - f_i|_{\Gamma_i} \|_{2}^{2} \leq \epsilon$ for all $i\in \N$ and for all $\epsilon > 0$.
\end{definition}

%The definition of concentration can be confusing at first sight. The definition is mainly of interest when $\epsilon$ is not too small, as the definition is trivially satisfied when $\epsilon \leq 1/|G|$ by choosing $P(s, t) = s$ and $\Gamma_i = \widehat{G}_i$. We will sometimes refer to a particular

Most applications are concerned with a single function that implicitly defines the entire family. In this case we informally say that the function, instead of the family, is concentrated. Examples of concentrated functions, and of this terminology, are given in Example \ref{example4}.

%A common abuse of notation is to refer to a single function as concentrated, in which case the interest is when $\epsilon \gg 1/|G|$ and $| \Gamma | \ll |G|$.
%In the literature it is usual to refer to a function as concentrated even though concentration is only defined for families of functions. For example a family of constant functions $\{f_n:\Z_n\to \C\}_{n\in \N}$ is concentrated (choose $\Gamma_n = \{\chi_{0}\}$ for every $n$) and as a short-hand one says that any constant function is concentrated.
%Note that when $p$ is very small then every function is concentrated, so this definition is most relevant for large $p$.

\begin{definition} [Heavy coefficient]
    For a function $f:G \rightarrow \C$ and a threshold $\tau>0$, we say that a coefficient $\widehat{f}(\alpha)$ (of the character $\chi_\alpha$) is \emph{$\tau$-heavy} if $|\widehat{f}(\alpha)|^2>\tau$.
\end{definition}

By Parseval's identity it is evident the number of $\tau$-heavy coefficients for a function $f:G \rightarrow \C$ is at most $\| f \|_2^2 / \tau$ (see \cite[Lemma 3.4]{KM} or \cite[Lemma 4.8]{Mansour94}). Thus, the cases of interest are where the latter value is polynomial in $\log(|G|)$, so there are at most polynomially many $\tau$-heavy coefficients. This forces $\tau$ to be relatively large to $\| f \|_2$, e.g. $\tau = \| f \|_2/ poly(\log(|G|))$. We remark that it might have been better to define a $\tau$-heavy coefficient to satisfy $|\widehat{f}(\alpha)|^2>\tau \| f \|_2^2$, however we keep the notion that is mostly used in the literature (as we show below most applications consider the specific case $\| f \|_2^2 = 1$).

The phrases \emph{significant coefficient} and \emph{heavy coefficient} are often used interchangeably to mean any coefficient $\widehat{f}(\alpha)$ which is large relative to the norm of the function, but without reference to any specific value of $\tau$.
% is not important, only that $\widehat{f}(\alpha)$ is large in size (specifically $\widehat{f}(\alpha)$ is $\tau$-heavy for some $\tau = poly(\log|G|)$). %We denote by $Heavy_\tau(f)$ the set of all $\tau$-heavy coefficients.
In this paper our convention is to use ``heavy'' in a formal sense and ``significant'' in an informal sense.

The relationship between concentrated functions and functions with significant coefficients is subtle.
If a function has a $\tau$-heavy coefficient, then it is $(1-\tau)$-concentrated (with $|\Gamma|=1$). But such a function is not necessarily $\epsilon$-concentrated for all $\epsilon$.
%a family of functions that all have a very heavy coefficient (e.g., of size half their norm) is not necessarily concentrated.
The literature has tended to focus on concentrated functions, but for many of the bit security applications it is sufficient that the function has one or more significant coefficients.
The distinction is important since it is harder to prove that a function is concentrated than to prove it has a significant coefficient.
%We emphasize that the notion of concentration is relevant for approximation of functions to any (not too small) error term $\epsilon$. For our applications we only need functions with significant coefficients.

\begin{example}\label{example4}
    Here are some examples of functions with significant coefficients, most of which are concentrated:
    \begin{itemize}
        \item A single character is concentrated; that is, the family $\{\chi_\alpha:\Z_n \to \C\}_{n>\alpha}$ for some $\alpha \in \N$ is concentrated. The case $\alpha = 0$ corresponds to constant functions, which are concentrated but will be un-interesting in our applications.

        \item For the least-significant-bit function $\LSB(x)$ on $\Z_{2^n}$, which gives the parity of $x$, the functions $f : \Z_{2^n} \to \C$ given by $f(x):=(-1)^{\LSB(x)}$ are concentrated.
              Indeed, these functions correspond to the characters $f(x) = (-1)^x = \omega_{2^n}^{2^{n-1} x} = \chi_{2^{n-1}}(x)$.

        \item The functions $\half :\Z_N \to \{-1,1\}$, for which $\half(x) = 1$ if $0 \leq x < \frac{N}{2}$ and $\half(x) = -1$ otherwise, are concentrated; one has $\widehat{\half}(\alpha) = \frac{1}{N} [ \sum_{0 \leq x < \frac{N}{2}} \overline{\chi_\alpha}(x) - \sum_{\frac{N}{2} \leq x < N} \overline{\chi_\alpha}(x) ]$.
              Elementary arguments (see Claim~\ref{fact1} below) show that
              \[
                  \left| \frac{1}{N}\sum_{0 \leq x < \frac{N}{2}}\overline{\chi_\alpha}(x) \right|
                  = \left| \frac{1}{N}\sum_{0 \leq x < \frac{N}{2}} \omega_N^{-\alpha x} \right|
                  < \frac{1}{| |\alpha|_N|}
              \]
              where $|\alpha|_N$ denotes the unique integer in $(-N/2, N/2]$ that is congruent to $\alpha$ modulo $N$.
              Similarly $\left|\frac{1}{N}\sum_{\frac{N}{2} \leq x < N}\overline{\chi_\alpha}(x)\right|<\frac{1}{| |\alpha|_N |}$.
              %, where $\|\alpha\|$ is the smallest distance from $\alpha$ to $N$.
              These results can be used to show that $\half$ is concentrated on a set of characters $\alpha$ with small $| |\alpha|_N |$; See~\cite[Claim 4.1]{AGS}. Similar arguments hold for the most-significant-bit function $f(x):=(-1)^{\MSB(x)}$, thus it is also concentrated.

        \item For primes $p$, the functions $f:\Z_p\to \C$ given by $f(x):=(-1)^{\LSB(x)}$ are concentrated.
              This follows from $f(x) = \half(2^{-1} x)$ and the scaling property.

        \item The function $\text{LPN}_s: \{0,1\}^n \to \{0,1\}$, given by $\text{LPN}_s(x) = (-1)^{\langle x,s \rangle + e(x)}$ for $e$ which is mostly $0$ (and otherwise $1$), has a significant coefficient and therefore is $\epsilon$-concentrated (for some large $\epsilon$).
              Let $I$ be the set for which $e(x) = 1$, then $\widehat{\text{LPN}_s}(s) = \frac{1}{2^n} \sum_{x \notin I} 1 + \frac{1}{2^n} \sum_{x \in I} (-1) = 1 - \frac{2|I|}{2^n}$. Since the size $|I|$ is relatively small, the coefficient $\widehat{LPN_s}(s)$ is large, that is, the function $\text{LPN}_s$ ``behaves'' like the character $\chi_s$ in $\{0,1\}^n$. If $|I|$ is very small, for example $|I|=poly(\log|G|)$, then $\text{LPN}_s$ is also concentrated. Moreover, one can show that $|\widehat{\text{LPN}_s}(v)| \le \tfrac{|I|}{2^n}$, and on average is expected to be proportional to $\sqrt{2|I|/ 2^n(2^n-1)} \approx \sqrt{2|I|}/2^n$.

              %as long as the function $e$ ``behaves randomly'' (BETTER PHRASE? DOES GAUSSIAN BEHAVE RANDOMLY? IS IT REALLY IMPORTANT THE e IS RANDOM; CAUSE WE DON'T HAVE SAMPLES HERE, THIS IS THE FOURIER TRANSFORM OVER THE ENTIRE DOMAIN) and the noise rate is not too high, for all $v \neq s$ it holds that $|\widehat{\text{LPN}_s}(v)| \approx \sqrt{\frac{|I|}{2^{2(n-1)}}}$. Notice that if $|I|$ is very small, for example $|I|=poly(\log|G|)$, then $\text{LPN}_s$ is also concentrated.

        \item `Noisy characters' given by $f(x):=\omega_p^{\alpha x+e(x)}$ for some suitable random functions $e$ have a significant coefficient $\widehat{f}(\alpha)$ as we show in Section \ref{sec:e-concentrated}. An example of such a noisy character is the function $\text{LWE}_s: \Z_p^n \to \Z_p$, given by $\text{LWE}_s(x) = \omega_p^{\langle x,s \rangle + e(x)}$ for $e(x)$ drawn from a Gaussian distribution.
    \end{itemize}
\end{example}

Another example of concentrated functions are the $i$-th bit functions, see Section~\ref{sec:i-th-bit} for details.

\subsection{Learning model}

Let $f: R \to \C$ be a function for which one wants to learn its significant coefficients.
The learner gets \emph{access} to samples of the form $(x, f(x))$. In the \emph{random access} model the learner receives polynomially many samples for inputs $x \in R$ drawn independently and uniformly at random. As opposed to this model, in the \emph{query access} model the learner can query the function on any chosen input $x \in R$ to receive the corresponding sample.

A \emph{learning algorithm} for a function $f:G\to\C$ outputs a set containing all the significant Fourier coefficients of $f$. Formally, given a function $f$ and $\epsilon, \delta > 0$, the algorithm outputs a set $\Gamma$ of size polynomial in $\log(|G|)$ and $\epsilon^{-1}$, such that $\| f - f|_\Gamma \|_{2}^{2} \leq \epsilon$ with probability at least $1-\delta$.

The main result of this subject (see Theorem~\ref{thm_sft_algo} below) is that there is a randomised polynomial-time algorithm to compute a sparse approximation $f|_\Gamma$ to a concentrated function in the query access model. In other words, concentrated functions admit a polynomial-time learning algorithm in the query access model.

\subsection{Probability}

The Chernoff bound gives an upper bound on the probability that a sum of independent random variables deviates from its expected value. One can therefore derive a lower bound for the number of samples needed to estimate the sum of independent random variables, with any required probability and error term. For a random variable $X$ on a set $A \subseteq \C$ we denote by $\E_{x \in A} X(x)$ the expected value $\sum_{x \in A} X(x) \Pr(x)$.

\begin{theorem}[Chernoff] \label{Chernoff}
    Let $A$ be a set of complex numbers such that $|x|\leq M$ for all $x\in A$. Let $x_i\in A$ be chosen independently and uniformly at randomly from $A$. Then
    \[
        \Pr\left[\left|\E_{x\in A}[x]-\frac{1}{m}\sum_{i=1}^mx_i\right| > \lambda \right]\leq 2 \mathrm{e}^{-\lambda^2 m / 2 M^2} .
    \]
\end{theorem}

%\begin{theorem}[Chernoff-Hoeffding] \label{Chernoff}
%   Let $A$ be a set of complex numbers such that $|x|\leq M$ for all $x\in A$. Let $\epsilon >0$. Let $m\geq \frac{2M^2}{\epsilon^2}\log(1/\delta)$. Let $x_i\in A$ be chosen randomly and uniformly from $A$ where $1\leq i \leq m$. Then
%   \[
%       \mathop{Pr}\left[\left|\E_{x\in A}[x]-\frac{1}{m}\sum_{i=1}^mx_i\right| > \epsilon \right]\leq \delta
%   \]
%\end{theorem}

\subsection{Table of notations}\label{sec:table}

We summarize the main notation and definitions in the following table.
\begin{center}
    \begin{tabular}{ l l }
        \hline
        Notation/Definition      & Meaning                                                                                \\ \hline
        $\omega_n$               & The complex $n$-th root of unity $\mathrm{e}^{2\pi i/n}$.                              \\
        $\chi$                   & A character of $G$.                                                                    \\
        $H^\perp$                & The orthogonal set $\{\alpha \in G \mid \chi_\alpha(h) = 1 \text{ for all }h\in H\}$.  \\
        $\widehat{f}$            & The Fourier transform of $f$.                                                          \\
        %    Time domain & The original domain on which a function $f$ is defined. \\
        %    Frequency domain & The domain of the Fourier transform $\widehat{f}$. \\
        Scaling property         & $\widehat{g}(\alpha) = \widehat{f}(c^{-1}\alpha)$ for $g(x) := f(cx)$ and $c \in R^*$. \\
        %    $f|_\Gamma$ & The restriction of $f$ to $\Gamma$, i.e. $\sum_{\chi_\alpha \in \Gamma}\widehat{f}(\alpha) \chi_\alpha$. \\
        $\tau$-heavy coefficient & A coefficient satisfying $|\widehat{f}(\alpha)|^2>\tau$.                               \\
        Significant coefficient  & A $\tau$-heavy coefficient, for some $\tau^{-1} = poly(\log|G| , \| f \|_\infty)$.     \\
        Query access             & The ability to ask for $f(x)$ for any input $x$.                                       \\
        \hline
    \end{tabular}
\end{center}

\section{Clarifications: Principles Underlying SFT Algorithms}
\label{sec_algorithm}

In the last few decades several significant Fourier transform (SFT) algorithms were proposed in the literature in several scientific areas. The early algorithms treat specific functions, while the later algorithms apply to classes of functions. The principles underlying these algorithms come from elementary group theory. The aim of this section is to clarify the rules that govern these algorithms. Our analysis gives a unified presentation for all of these algorithms, which we believe brings clarity to the literature and will be more accessible to non-experts.

A precise statement of what an SFT algorithm does is given in Theorem~\ref{thm_sft_algo}.
Section \ref{sec:history} gives an overview of the earlier algorithms.
Section~\ref{sec:SFTalgo} presents the unified SFT algorithm in the query access model. The section starts with a high-level presentation of the SFT algorithm.
We then describe the algorithm with a focus on the required algebraic relations between the queries, thus explaining the need for query access.
For these relations to arise, the function's domain needs to be ``highly composite'', i.e. to contain many subgroups. We give examples of the requirement on the queries in some specific domains. We then turn to an analysis of the algorithm on domains of prime order. Moreover, the original approach that we present in Section~\ref{sec:mod-switch} gives further insights on the connections between the different domains. We finish this section with two short descriptions. Section~\ref{sec:noisy-oracle} discusses cases where some of the function's outputs (i.e. the algorithm's inputs) are ``noisy'', that is where the actual values are replaced with some other values. Section~\ref{sec:hardness-LPN} explains why an SFT algorithm in the random access model is unlikely to exist.

\begin{theorem} [{\cite[SFT algorithm]{AKA-thesis}}{\cite[Theorem 5]{AGS}}]
    \label{thm_sft_algo}
    Let $G$ be an abelian group represented by a set of generators of known orders. There is a learning algorithm that, given query access to a function $f: G \rightarrow \C$, a threshold $\tau > 0$ and $\delta > 0$, outputs a list $L$ of size at most $2\|f\|_2^2 / \tau$ such that
    \begin{itemize}
        \item $L$ contains all the $\tau$-heavy Fourier coefficients of $f$ with probability at least $1-\delta$;
        \item $L$ does not contain coefficients that are not $(\tau/2)$-heavy with probability at least $1-\delta$.
    \end{itemize}
    The algorithm runs in polynomial time in $\log\left(|G|\right)$, $\| f \|_\infty^2 / \tau$ and $\log\left(\frac{1}{\delta}\right)$.
\end{theorem}

\subsection{History and special cases}
\label{sec:history}

Key ideas behind the SFT algorithm first arose in other settings, and the aim of this section is to put some of this early work in context. This section is not needed in order to understand the SFT algorithm. Readers who are mainly interested in understanding the general SFT algorithm should feel free to skip this section and go straight to Section~\ref{sec:SFTalgo}.

\subsubsection{Goldreich--Levin} \label{sec:GL}

Consider a `noisy' inner product function $f_s: \{0,1\}^n \to \{0,1\}$ given by $f_s(x) = \langle x, s \rangle + \delta(x)$ (addition takes place mod $2$) where $\delta(x) = 1$ with some small probability (noticeably smaller than $1/2$) and otherwise $\delta(x) = 0$.
This is the same function as in the well-known learning parity with noise (LPN) problem.
The task is to learn $s$ given samples $f_s(x_i)$.

The connection to the Fourier basis can be seen by reformulating the problem as follows. Define $g: \{0,1\}^n \to \{-1,1\}$ by $g(x) = (-1)^{f_s(x)} = (-1)^{\langle x, s \rangle + \delta(x)}$. Notice that when $\delta \equiv 0$ then $g$ is in fact the character $\chi_s(x) = (-1)^{\langle x, s \rangle}$. The fact that $\delta(x)=0$ on most inputs guarantees that $\widehat{g}(s)$ is a significant Fourier coefficient for $g$, as shown in Example~\ref{example4}.

In the random access model, where one gets arbitrary samples, LPN is considered to be a hard computational problem (unless $\delta(x) = 0$ for all, or almost all, $x$; then reconstructing $f_s$ is an easy linear algebra problem).
Goldreich and Levin~\cite{GL} (GL) considered this problem in the query access model, and gave an efficient algorithm to solve it as we briefly explain. In the simplest setting there is a single $\tau$-heavy coefficient for $\tau > 1/2$.

If one can choose the queries for $f_s$ then an elementary approach is to query on the unit vectors $e_1 := (1,0,\dots,0), \dots, e_n := (0,\dots,0,1)$ to learn $s$ bit-by-bit. However, since the query on $e_i$ may return the answer $\langle e_i, s \rangle + 1$, one would like to generate a small set of independent values of the form $\langle e_i, s \rangle + \delta$, and determine $s_i$ by majority rule, as $\delta=0$ with probability noticeably greater than $1/2$. This can simply be achieved by querying on correlated values $x$ and $x + e_i$ to get the results $\langle x, s \rangle + \delta(x)$ and $\langle x, s \rangle + \langle e_i, s \rangle + \delta(x+e_i)$. If both answers are not noisy (or if both are noisy) then by subtracting one from the other we get $\langle e_i, s \rangle$, which is the $i$-th coordinate of $s$.
(For the interested reader: if the noise rate is at least $1/4$, then there may not be a unique solution (see Section \ref{sec:noisy-oracle}); Rackoff (see~\cite[Section C.2]{Rackoff}) suggested to use a trick due to Alexi et al.~\cite{ACGS} to deal with this case.)

The original Goldreich--Levin paper~\cite{GL} does not give a clear description of the learning algorithm. A description in the language of Fourier analysis was given in~\cite{KM} by Kushilevitz and Mansour.

%\begin{remark}
%    \textbf{(***) Comments: This is only correct if the algorithm is non-adaptive. Also, is this in the right place? I presume it is here to explain why some authors speak of ``list decoding''.}
%
%    Note that the problem of learning $s$ from noisy inner products $\langle x, s \rangle$ can be interpreted as decoding a binary linear code.
%    The choosen queries $x$ can be collected as the generator matrix for the code.
%    The SFT algorithm can therefore be viewed as a decoding algorithm.
%    In the situation where the error rate is very high and there is not a uniquely determined solution then the SFT algorithm can be viewed as a list-decoding algorithm. More on the relation of these tools to decoding linear codes can be found in the recent work \cite{Ext-KM}, where an `extended KM' algorithm is presented.
%    % \\\textbf{XXX Steven: any further explanation? XXX}
%\end{remark}

\subsubsection{Bleichenbacher}

Bleichenbacher~\cite{Bleichenbacher} seems to have been the first to consider these problems in the case of functions on $\Z_N$ where $N$ is not a power of 2. He considers a `noisy' product function $f_s: \Z_N \to \Z_N$ given by $f_s(x) = sx + \delta(x)$ where $|\delta(x)| < \frac{N}{2^\lambda}$, for some real number $\lambda$, with probability (noticeably) greater than $1/2$. This is usually viewed as outputting about $\lambda$ most significant bits of the product $s x \in \Z_N$, as $s x$ and $f_s$ differ by a small number. The task, as before, is to learn $s$ given samples $f_s(x_i)$. The connection to the Fourier basis can be seen by reformulating the problem as done in the previous GL case -- see the `noisy character' case in Example~\ref{example4}.

This problem is in fact the hidden number problem that was considered in \cite{BV} and which we further discuss in Section~\ref{sec_applications}. Notice that if one can obtain any query, then this problem can be solved by successively multiplying by $2$ to read the bits of $s$. Since some samples may be erroneous, majority rule is used, similar to the approach taken in the GL case. Moreover if $\delta(x)$ is very small, finding $s$ and reconstructing $f_s$ is easy (by ranging over all possible values for $\delta$).

Bleichenbacher's original setting takes place in the random access model, so he gives a method (not efficient for large domains) to obtain samples $f_s(x)$ for which $x$ lie in short intervals, and then gives a method to solve the original problem. We explain the latter method.
Here however, one is not assumed to have any chosen query, but only that the queries lie in some (designated) intervals.

The main idea to solve this problem comes from the fact that one can use small (but gradually increasing) multipliers, not necessary powers of $2$, to learn the bits of $s$. This comes from the following observation: if $s < \frac{N}{2^\eta}$, for some $\eta \geq 0$, then $s y < N$ for every $0\leq y \leq 2^\eta$. In other words, the product $sy$ does not `wrap-around' the modulus $N$.

The latter observation can be used to determine upper bits of $s$: given $y$ and $f_s(y) = sy+\delta(y)$, take $\lfloor f_s(y)/y \rceil = \lfloor s + \delta(y)/y\rceil$; assuming there is no wrap-around over $N$ in $f_s(y)$, we get some upper bits of $s$. For example, if $2^{\eta-1}\leq y \leq 2^\eta$, then $|\delta(y)/y| < \frac{N}{2^{\eta+\lambda-1}}$ so we roughly learn $\lambda-1$ of the upper bits of $s$ that were not already known.

Now suppose one knows $\MSB_\rho(s)$, the $\rho$ most significant bits of $s$, then by subtracting it from $s$ we have $s' := s-\MSB_\rho(s) < \frac{N}{2^\rho}$. The goal now is to learn further (upper) bits of $s'$. One can define $f_{s'}(y)$ to be $f_s(y) - \MSB_\rho(s) y = sy - \MSB_\rho(s)y + \delta(y) = s' y + \delta(y)$. Thus, for appropriate multiplier $y$, say $2^{\rho-1}\leq y \leq 2^\rho$, we can determine more upper bits of $s'$ as above. Repeating this procedure, one eventually learns all bits of $s$.

Notice that this approach requires having multipliers drawn from some interval  $\{0, 1, \ldots, 2^i - 1\}$ (specifically small multipliers in the first stages, which are the `hardest' to get). Moreover, since it is not always the case that $|\delta(x)| < \frac{N}{2^\lambda}$, we need to generate independent multipliers from these intervals. Similar to the approach in the GL case, this is done by fixing some $z$ and querying on $z+r$ for $r$ chosen uniformly in $\{0, 1, \ldots, 2^i - 1\}$, then subtracting. Thus the queries have to be correlated such that their difference lies in the required interval.

This description presents the core ideas behind Bleichenbacher's algorithm in a manner similar to the description of the GL algorithm above. Bleichenbacher's description, which involves terminology from Fourier analysis, resembles the Kushilevitz--Mansour modification to the GL algorithm (see below) and the ideas described in Section~\ref{sec:mod-switch}. For the full details we refer to Bleichenbacher~\cite{Bleichenbacher} (see also Section~\ref{sec:e-concentrated} below).
This method does not seem to have been used for cryptographic applications until the recent works~\cite{De-Mulder,GLV/GLS}.

\subsubsection{Following work}

The early work did not explicitly mention Fourier coefficients, but it was realised that one can re-phrase the problems as finding significant Fourier coefficients of related functions, as we show above.
The Goldreich--Levin case was generalized by Kushilevitz and Mansour~\cite{KM} (KM) to any real-valued function over $\{0,1\}^n$ and this work was the first to explicitly treat functions with more than one significant Fourier coefficient.

Subsequently, Mansour~\cite{Mansour92} gave an algorithm for functions $f : \Z_{2^n} \to \C$. Unlike other works, Mansour's algorithms computes the significant coefficients from the least significant bit to the most significant bit (a link between these works~\cite{KM,Mansour92} is explained in Remark \ref{rem_mod-dim} below).
The approach of Mansour was extended, thereby giving a generalisation of Bleichenbacher's result, by Akavia, Goldwasser and Safra~\cite{AGS} (AGS).
% for any complex-valued function over $\Z_N$.
%The work of Mansour~\cite{Mansour92} is a hybrid of the two, which works over $\Z_{2^n}$ and finds the heavy coefficients bit-by-bit from the least-significant bit to the most-significant bit. We show below that one can take the ideas of Mansour and apply them directly to $\Z_p$.

Notice that combining the KM and AGS ideas gives an algorithm for all groups $\Z_{N_1} \times \dots \times \Z_{N_r}$, since one can easily collapse from the latter to $\Z_{N_j}$ (by choosing appropriate queries, for example queries of the form $r \cdot e_j$ for desired values $r \in \Z_{N_j}$). Therefore, the case of most interest is $G = \Z_p$ which we present below.
As further evidence for the unity of all these ideas we remark that the KM and AGS algorithms query on exactly the same set of queries as GL and Bleichenbacher (and subsequently reveal the significant coefficients bit-by-bit from MSB to the LSB).

\subsection{The SFT algorithm}\label{sec:SFTalgo}

Let $f: G \to \C$. Given a threshold $\tau \in \R$, the algorithm outputs all $\tau$-heavy Fourier coefficients of $f$ (and potentially some other $\tau/2$-heavy coefficients) with overwhelming probability.

We first give a high-level view of how the algorithm works. The method is a form of binary search: the algorithm divides the set of Fourier coefficients into two (disjoint) sets, say $A$ and $B$, and checks each set separately to determine whether it potentially contains a $\tau$-heavy coefficient. To do this the algorithm defines two new functions, one for each set of coefficients.
A clever use of Parseval's identity allows the algorithm to check the size of all coefficients in each set simultaneously, given the norm of each function.
Hence, the task is to determine the norms of the two new functions, which requires a method to compute the function outputs.
The structure of the sets $A,B$ is important: for some sets we have useful formulas to compute the functions at required values.
Instead of precisely calculating these values, it is sufficient to have approximations of the outputs of the functions and to approximate the norm of each function. The Chernoff bound is then used to bound the error term in the approximations.

Schematically, the algorithm operates as follows, where we initially take $D = G$:
\begin{itemize}
    \item Partition $D = A \cup B$, and define $f_A(x) := \sum_{\alpha \in A} \widehat{f}(\alpha)\chi_\alpha(x)$ and $f_B(x) := \sum_{\beta \in B} \widehat{f}(\beta)\chi_\beta(x)$.
    \item Approximate the values $f_A(x_i)$ and $f_B(y_j)$ for polynomially many samples $x_i, y_j$, chosen uniformly at random. This is done using the fundamental relation in (\ref{eq:conv}) below.
    \item Using the values from the previous step, approximate the norms $||f_A||_2^2$ and $||f_B||_2^2$. See (\ref{eq_STF}).
    \item Using Parseval's identity $||f_A||_2^2 = \sum_{\alpha \in A} |\widehat{f}(\alpha)|^2$, if the approximation of the norm is smaller than\footnote{A lower threshold $\frac{3}{4}\tau$ is needed since the algorithm only approximates the norm. As a consequence, the final list may contain coefficients that are $\frac{\tau}{2}$-heavy but not $\tau$-heavy.} $\frac{3}{4}\tau$ then with overwhelming probability $f$ does not have a $\tau$-heavy coefficient in $A$. Hence, dismiss $A$. Act similarly for $f_B$.
    \item Run the algorithm recursively on the remaining sets and stop when it reaches singletons.
\end{itemize}
% As already explained, Parseval's identity shows that a function can only have polynomially many significant coefficients. Hence, by setting a threshold $\tau$ such that $\tau^{-1} = poly(\log|G| , \| f \|_\infty)$, the number of sets involved in the process (therefore the number of iterations) is polynomial (see \cite[Lemma 3.4]{KM} or \cite[Lemma 4.8]{Mansour94})\footnote{The notion of heavy coefficient in \cite{KM,Mansour94} is slightly different from ours. There, the algorithm outputs coefficients satisfying $|\widehat{f}(\alpha)| > \tau$.}.

\begin{remark}
    We emphasize that the algorithm can work with any function $f$ and with any threshold $\tau$. Specifically, if $f$ does not have any $\tau$-heavy coefficients, then the algorithm will output an empty list.
    However, the running time is polynomial in $\|f\|_\infty^2/\tau$ so the algorithm will not be efficient if the threshold is chosen to be too low.
    % Indeed, if $\tau$ is small then the algorithm will insist on using many samples to get sufficiently close approximations to the norms.

    % To illustrate these points, consider a function that has a coefficient $\widehat{f}(v)$ that is relatively large compared to each of the other coefficients but is not significant (for example 10 times larger than each of the rest), i.e. it is not large relative to the norm. Suppose one tries to find this coefficient by setting a very low threshold. The running time of the algorithm on this function would not be polynomial as, at the first stages, the sums of all Fourier coefficients over the sets are roughly the same size. Therefore, the algorithm would have to keep all the sets until they are sufficiently small. This case corresponds to a $\tau$-heavy $\widehat{f}(v)$ for $\tau^{-1}$ that is not polynomial in $\log|G|, \| f \|_\infty$. To get that $\tau^{-1}$ is polynomial in $\log|G|,\| f \|_\infty$, the size of $\widehat{f}(v)$ should be comparable with the norm of the function and not just larger than all the rest.
\end{remark}

\subsubsection{Domains of size $2^n$}

We now sketch an algorithm that unifies the KM and Mansour algorithms. Our presentation is more group-theoretic than the original works.
We refer to~\cite{KM} and ~\cite{Mansour94} for exact details and proofs.

Let $f: G \to \C$ and $\tau \in \R$.
% Recall that $f(x) = \sum_{\alpha \in G} \widehat{f}(\alpha) \chi_\alpha(x)$.
At each iteration the algorithm takes a set $D$ (starting with $D=G$) and proceeds as follows.

\noindent\textbf{Partial functions.} Partition $D = A \cupdot B$ into two sets that are defined below.
%
%The following shows how to determine if the the set $A$ does not contain $\tau$-heavy coefficients, that is, if there is no $\alpha \in A$ such that $\widehat{f}(\alpha)$ is $\tau$-heavy. Similarly, it holds for $B$.
Define the function $f_A: G \to \C$ by $f_A(x) = \sum_{\alpha\in A} \widehat{f}(\alpha)\chi_\alpha(x)$.
If $f$ has a $\tau$-heavy coefficient $\alpha$ and $\alpha \in A$, then $f_A$ has a $\tau$-heavy coefficient. All arguments hold similarly for the set $B$.

\noindent \textbf{Estimating $f_A(x)$.} We need a method to estimate values of the function $f_A$ using values of the original function $f$.
%This is done by ``filtering" the coefficients of the function $f$, using the query access, and keeping only those which are in $A$.
We define a \emph{filter function} $h_A: G \to \C$ by $h(x) = \sum_{\alpha \in A} \chi_\alpha(x)$, and then use the property $\widehat{f \ast h_A} = \widehat{f} \cdot \widehat{h_A}$.
Since
\[ \widehat{h_A}(\alpha) = \left\{
    \begin{array}{l l}
        1 & \quad \alpha \in A,       \\
        0 & \quad \mathrm{otherwise,}
    \end{array}
    \right. \]
we have
\[\widehat{f \ast h_A}(\alpha) = \left\{
    \begin{array}{l l}
        \widehat{f}(\alpha) & \quad \alpha \in A,       \\
        0                   & \quad \mathrm{otherwise.}
    \end{array}
    \right.\]
In other words,
\begin{equation}\label{eq:conv}
    f \ast h_A = f_A \, .
\end{equation}
%the function $f \ast h_A$ is the restriction of $f$ to $A$, i.e. .
%
%\footnote{Note that since we are only concerned with the size $|\widehat{f}(\alpha)| = |\widehat{f \ast h_A}(\alpha)| / |\widehat{h_A}(\alpha)|$ it is possible to define other filter functions $h$ satisfying $|\widehat{h_A}(\alpha)| = 1$ for all $\alpha \in A$. As shown later, it is sufficient that $|\widehat{h_A}(\alpha)| \approx |\widehat{h_A}(\alpha')|$ for all $\alpha,\alpha' \in A$.}

Convolution is not a task we have an efficient method to calculate in general, let alone efficiently calculating $h_A(x) = \sum_{\alpha \in A} \chi_\alpha(x)$. Therefore, the structure of the sets is important and plays a key role in the ability to apply the algorithm. Notice that if $A$ is an arithmetic progression, then $\sum_{\alpha \in A} \chi_\alpha(x) = \sum_j \chi_{qj+r}(x) = \chi_r(x) \sum_j \chi_q(jx)$, and so it can be evaluated by the formula for geometric series. More generally, assume $D \leq G$ is a subgroup and let $H \leq D$ be a subgroup (of index $2$). We take $A$ to be a coset $A = z+H$ for some $z\in G$ (then $B$ is taken to be the other coset). Then,
\[ h_{z+H}(x) = \sum_{h\in H}\chi_{z+h}(x) =\sum_{h\in H}\chi_z(x)\chi_h(x) = \chi_z(x)\sum_{h\in H}\chi_h(x) \,,\]
and the latter is zero unless $x \in H^\perp$ ($H^\perp$ is defined in (\ref{eq_Hperp}) above). Thus the function $h_A$ is given by
\begin{equation} \label{filter}
    h_A(x) = h_{z+H}(x) = \begin{cases} \chi_z(x)\cdot |H|, & \text{if }x\in H^\perp, \\ 0, &\text{otherwise.}\end{cases}
\end{equation}

We therefore get, since $|H| |H^\perp| = |G|$,
\begin{equation*}
    \begin{split}
        f_A(x) & = f \ast h_A(x) = \E_{y \in G}\left[f(x-y)h_A(y)\right] = \frac{1}{|G|} \sum_{y \in G} f(x-y)h_A(y) \\
        & = \frac{1}{|G|} |H| \sum_{y \in H^\perp} f(x-y)\chi_z(y) = \E_{y \in H^\perp}\left[f(x-y)\chi_z(y)\right] .
    \end{split}
\end{equation*}

% The last term is an expectation over values of size at most $\| f \|_\infty$, and so the Chernoff bound guarantees that polynomial (in $\log(|G|)$) many samples (chosen uniformly at random in $H^\perp$) are sufficient to approximate it with an error term of size at most $\frac{\|f\|_\infty}{poly(\log(|G|))}$ with overwhelming probability.\\
% We give concrete examples of this step in Section \ref{sec_examples} below.

\noindent \textbf{Estimating $\| f_A \|_2$.} We can now write $\|f_A\|^2$ as
\[ \| f_A \|_2^2 =  \E_{x\in G}\left|(f * h_A)(x)\right|^2 = \E_{x\in G}\left|\E_{y\in G}\left[f(x-y)h_A(y)\right]\right|^2 = \E_{x\in G}\left| \E_{y\in H^\perp}\left[f(x-y)\chi_z(y)\right]\right|^2 . \]
Again, an approximation of the norm is sufficient (a consequence of the approximation is that we have to lower the threshold $\tau$ a little bit).

We can therefore approximate $\| f_A\|_2^2$ by choosing $m_1, m_2$ sufficiently large (given by the Chernoff bound), randomly choosing\footnote{Note that as in \cite{KM,Mansour92} one can define the function $f_A$ over $H$ (and not $G$), and therefore choose the values $x_i$ from $H$.} $x_i\in G$ where $1\leq i \leq m_1$, randomly choosing $y_{ij}\in H^\perp$ for each $i$ where $1\leq j \leq m_2$ and calculating
\begin{equation}\label{eq_STF}
    \frac{1}{m_1}\sum_{i = 1}^{m_1}\left|\frac{1}{m_2}\sum_{j=1}^{m_2}f(x_i-y_{ij})\chi_z(y_{ij})\right|^2 \approx \| f_A\|_2^2 = \sum_{\alpha \in A} |\widehat{f}(\alpha)|^2 .
\end{equation}

One then checks if this value is smaller than $3\tau /4$. If so then with overwhelming probability there is no $\alpha \in A$ such that $\widehat{f}(\alpha)$ is $\tau$-heavy, and so the set $A$ can be dismissed. Notice that if this value is greater than $3\tau /4$ it does not necessarily mean that $A$ contains a significant coefficient. In this case the algorithm sets $D=A$ and repeats until all sets are singletons or dismissed.
% As mentioned above, as long as the threshold $\tau$ satisfies $\tau^{-1}=poly(\log|G|, \| f \|_\infty)$, it is guaranteed that the number of sets the algorithm keeps throughout the process is polynomial in $\log(|G|)$.

We give the pseudocode of the algorithm. At start, set $z=0$ and $k=n$, so $H_k=G$.
\begin{center}
    \begin{algorithm}[H]
        \caption{MainProcedure}
        \DontPrintSemicolon
        \KwIn{A coset $z+H_k$.}% where $z\in G$ and $1\leq k \leq n$.}

        \eIf { $|H_k| = 1$ } {
        \eIf{$|\noit{Est}\widehat{f}(z)|^2\geq 3\tau/4$}{
            \Return{$\{z\}$} \;
        } {
            \Return{$\emptyset$} \;
        }
        } {
        Let $W$ be a set of coset representatives for $H_{k-1}$ in $H_k$ \;
        Let $W' = \{w\in W \mid  \noit{EstNormSq} (f_{(z+w)+H_{k-1}}) \geq  3\tau/4 \}$  \;

        \Return{$\cup_{w\in W'}\noit{MainProcedure}((z+w)+H_{k-1})$ } \;
        }
    \end{algorithm}
\end{center}
\begin{center}
    \begin{algorithm}[H]
        \caption{EstNormSq \label{EstimateNorm}}
        \DontPrintSemicolon
        \KwIn{$f_{z+H}:G\to \C$.}
        Choose $x_i\in G$ where $1\leq i \leq m_1$ \;
        For each $i$, choose $y_{ij}\in H^\perp$ where $1\leq j \leq m_2$ \;
        \Return{$\frac{1}{m_1}\sum_{i = 1}^{m_1}\left|\frac{1}{m_2}\sum_{j=1}^{m_2}f(x_i-y_{ij})\chi_z(y_{ij})\right|^2$}
    \end{algorithm}
\end{center}\begin{center}
    \begin{algorithm}[H]
        \caption{Est$\widehat{f}$ \label{EstimateCoefficient}}
        \DontPrintSemicolon
        \KwIn{$z\in G$.}
        Choose $x_i\in G$ where $1\leq i \leq m_1$ \;
        \Return{$\frac{1}{m_1}\sum_{i = 1}^{m_1}f(x_i)\chi_z(-x_i)$}
    \end{algorithm}
\end{center}

\subsubsection{Examples}\label{sec_examples}

Notice that in (\ref{eq_STF}) above for each $x_i$ one needs the samples $f(x_i-y_{ij})$. This explains the importance of having query access to the function. To illustrate this point, we give some concrete examples.

Kushilevitz and Mansour \cite{KM} consider a function $f: \{0,1\}^n \to \R$. Write $x = x_1 \ldots x_n$. At the first iteration define $A$ to contain all $n$-bit strings that start with $0$ and $B$ to contain all the $n$-bit strings that start with $1$. Then we have
\begin{equation}\label{eq_KM}
    h_A(x) = \begin{cases} 2^{n-1}, & \text{if }x = 0\ldots0 \text{ or } x=10\ldots0, \\ 0, &\text{otherwise,}\end{cases}
\end{equation}
and indeed
\[\widehat{h_A}(\alpha) = \frac{1}{2^n}\sum_x h_A(x) (-1)^{\langle \alpha, x \rangle} = \frac{1}{2} \left( (-1)^0 + (-1)^{\alpha_1} \right) = \left\{
    \begin{array}{l l}
        1 & \quad \alpha \in A;       \\
        0 & \quad \mathrm{otherwise.}
    \end{array}
    \right.\]
One can only evaluate $f \ast h_A(x)$ if one has the values $f(x)$ and $f(x+e_1)$.
This shows that the KM approach requires (in the first iteration) queries on pairs of vectors that differ by a unit vector, exactly as in the elementary approach to the GL theorem as sketched in Section~\ref{sec:GL}.

Mansour \cite{Mansour92} considers a function $f: \Z_{2^n} \to \C$. At the first iteration define $A$ to contain all the even numbers in $\Z_{2^n}$ and $B$ to contain all the odd numbers. Then, we have
\begin{equation}\label{eq_Mansour}
    h_A(x) = \begin{cases} 2^{n-1}, & \text{if }x = 0 \text{ or } x=2^{n-1}, \\ 0, &\text{otherwise,}\end{cases}
\end{equation}
and indeed
\[\widehat{h_A}(\alpha) = \frac{1}{2^n}\sum_x h_A(x) \omega_{2^n}^{\alpha x} = \frac{1}{2} \left( 1 + (-1)^{\alpha} \right) = \left\{
    \begin{array}{l l}
        1 & \quad \alpha \in A;       \\
        0 & \quad \mathrm{otherwise.}
    \end{array}
    \right.\]
One can only evaluate $f \ast h_A(x)$ if one has $f(x)$ and $f(x+2^{n-1})$.

The analysis of this algorithm is useful for the prime case below, and so we present its later stages. In stage $l$ of this algorithm, one defines the subgroup $H$ to contain all multiples of $2^l$ in $\Z_{2^n}$. Hence the cosets used to partition the solution space contain all numbers that agree on their remainder modulo $2^l$, and $H^\perp = \{x \in \Z_{2^n} \ | \ x 2^l \equiv 0 \pmod{ 2^n } \} =  \{0, 2^{n-l}, 2 \cdot 2^{n-l}, 3 \cdot 2^{n-l}, \dots, (2^l-1) 2^{n-l} \}$. Define $A = A_r = \{x \in \Z_{2^n} \ | \ x \equiv r \pmod{2^l} \} = H + r$. Then, the filter function $h_A$ satisfies
\begin{equation}\label{Zp_filter}
    h_A(x) = \begin{cases} \chi_{r}(x) \cdot 2^{n-l}, & \text{if }x\in H^\perp, \\ 0, &\text{otherwise.}\end{cases}
\end{equation}
Again, to approximate $f \ast h_A(x)$, one needs enough samples $f(x_i)$ for $x_i \in H^\perp$.

\begin{remark}\label{rem_mod-dim}
    Readers familiar with lattice cryptography may be interested to know that %one can view the relationship between the KM~\cite{KM} algorithm on $\{0,1\}^n$ and the Mansour~\cite{Mansour92} algorithm on $\Z_{2^n}$ as the modulus-dimension tradeoff~\cite{mod-dim}.
    the idea that underlies the modulus-dimension tradeoff~\cite{mod-dim} already appears in the relationship between the KM~\cite{KM} algorithm on $\{0,1\}^n$ and the Mansour~\cite{Mansour92} algorithm on $\Z_{2^n}$.
    We briefly sketch this idea. Let $\textbf{a} = (a_0,\ldots,a_{n-1}) \in \Z^n_p$, $\textbf{s} = (s_0,\ldots,s_{n-1}) \in \{0,1\}^n$, and suppose
    \[
        b \equiv \textbf{a} \cdot \textbf{s} + e \equiv \sum_{i=0}^{n-1} a_i s_i \ + \ e \pmod{p} \,.
    \]
    Writing $a = a_0 p^{n-1} + a_1 p^{n-2} + \cdots + a_{n-2} p + a_{n-1}$ and $s = s_0 + s_1 p + \cdots + s_{n-1} p^{n-1}$ we have
    \[
        as \equiv (a_0 s_0 + \cdots + a_{n-1} s_{n-1})p^{n-1} + \text{lower term} \pmod{p^n}
    \]
    and some of its MSBs agree with the MSBs of $b p^{n-1}$, when $p$ is large.

    As shown in equation~(\ref{eq_KM}) above, at the first iteration over $\{0,1\}^n$ the filter function is nonzero on the inputs $\textbf{0}$ and $\textbf{a} = (1,0,\ldots,0)$ in $\Z^n_2$.
    These vectors correspond to the values $a = 0$ and $a = 2^{n-1}$ in $\Z_{2^n}$, which are exactly the values appearing in equation~(\ref{eq_Mansour}). Since the lower terms of $a \cdot s$ are zero, when $a=0,p^{n-1}$, the MSB of $as$ and $b p^{n-1}$ agree even for $p=2$. In both domains, we use these values to recover $s_0$.
    %
    %iteration) over $\Z_{2^n}$ is nonzero \textendash \space and indeed we use these values to find $s_0$, which is the least-significant bit of $s$ ($s$ mod $2$), as the original inputs in $\{0,1\}^n$ allows us to determine $s_0$ (the first coordinate of $\textbf{s}$): notice that for these values we have an equality $a' \cdot s = \langle \textbf{a}, \textbf{s} \rangle \cdot 2^{n-1}$.
    The generalization to all inputs $\textbf{a}$ arising in the algorithms is straightforward.
\end{remark}

\subsubsection{Domains of prime order}

The ideas behind the algorithm presented above make use of the fact that the domain's order can be factored as a product of small primes (especially for powers of $2$, as been shown for $\{0,1\}^n$ in \cite{KM} and for $\Z_{2^n}$ in \cite{Mansour92}). A case of interest, from the theoretical and practical sides, is domains of (large) prime order. Notice that each additive group $\Z_N$ can be decomposed into a direct product of prime subgroups $\Z_{p_1} \times \dots \times \Z_{p_n}$. The query access allows us to work over each subgroup separately, to recover the coefficients prime-by-prime, similar to the bit-by-bit approach in the GL case above (Section~\ref{sec:GL}). Indeed one can query on $\textbf{x}\boldsymbol{_j} = (0,\dots,0,x_j,0,\ldots,0)$ to work over the group $\Z_{p_j}$.\footnote{Since deterministic queries are not desirable, additional randomization is used in practice.} Thus being able to find heavy coefficients for functions over a prime group $\Z_p$ will allow us to find heavy coefficients for functions over any $\Z_N$.

For prime groups the analysis we presented for the algorithm above does not apply as $\Z_p$ does not have any proper subgroups, specifically not those of small index.
The importance of the subgroups is in the evaluation of exponential sums (such as equation~(\ref{eq:Hexp}) above), which subsequently allows us to have useful formulas for the filter functions (such as equation~(\ref{filter})).
We now show that one can still follow the steps in the algorithm above. Natural candidates for the partitioning sets are intervals (of similar size) of consecutive numbers or classes of numbers with the same remainder modulo $2^l$ (where $l$ represents the stage we work at), which is similar to the approach taken over $\Z_{2^n}$ (see Section \ref{sec_examples}).\footnote{Note that both are arithmetic progressions, which allow evaluating $h_A$.} In fact, using the frequency-shifting and scaling properties of the Fourier transform, one can show that these two partitions are equivalent (where there is a correspondence between the size of the intervals and the size of the classes), in the sense that one can transform the coefficients in an interval to coefficients of the same class modulo $2^l$ and vice versa. We show this equivalence below.

The algorithm over $\Z_p$~\cite{AGS}
works in the same steps as explained in Section \ref{sec:SFTalgo}. The main obstacle is to show how to efficiently calculate the function $f_A$, for some appropriate set $A$. We therefore focus on this step. The other steps are similar to the algorithm for domains of size $2^n$.

\noindent \textbf{Working in the `frequency domain'.} In order to show the difficulty working in a domain of prime size, we start with a naive imitation of the approach taken in the algorithm for domains of size $2^n$. Let $A$ be an arithmetic progression in $\Z_p$, and define $f_A = \sum_{\alpha\in A} \widehat{f}(\alpha)\chi_\alpha(x)$ and $h_A(x) = \sum_{\alpha \in A}\chi_\alpha(x)$ as above. Then $f_A(x) = f * h_A(x) = \E_{y \in G}\left[f(x-y)h_A(y)\right] = \E_{y \in G}\left[f(x-y)\sum_{\alpha \in A}\chi_\alpha(y)\right]$. Since $A$ is an arithmetic progression, $\sum_{\alpha \in A}\chi_\alpha(x)$ is a geometric progression for which we have a formula. We get that $f_A(x)$ is an expectation over values each of which we can calculate exactly. Moreover, unlike in the algorithm above, the filter function here is nonzero over a very large set, and therefore one can hope that specific queries are not needed in this case (as shown in Section~\ref{sec_examples} the previous filter functions are zero almost everywhere, so in order to get a good approximation of $f_A(x)$ we need the specific inputs where the filter function is not zero). This turns out to be a disadvantage. Indeed, in order to determine $f_A(x)$ in polynomial time, we can only approximate this expectation, but as the values of this geometric progression can be as large as $|A|$, one derives from the Chernoff bound that the number of samples needed to have a good approximation of $f_A(x)$ is roughly $|A|$, which is exponential in $\log(p)$ in the first stages of the algorithm.
Hence this approach is not practical.

\noindent \textbf{Working in the `time domain'.} Instead of working in the `frequency domain', we can work in the `time domain'. In this case we define $A$ to be a class of numbers with the same remainder mod $2^l$. We adapt the filter function in (\ref{Zp_filter}) to the $\Z_p$ case. As in Section \ref{sec_examples}, let $H$ be the set containing all multiples of $2^l$ in $\Z_p$. Define $H^\perp := \{0, 2^{-l}, 2 \cdot 2^{-l}, \dots, (2^l-1) 2^{-l} \}$. Notice that while $H^\perp$ is not orthogonal to $H$, it contains all numbers that give small remainder (mod $p$) when multiplied by $2^l$. Let $z \in \Z_p$ such that $z \equiv r \pmod{2^l}$ and define $A=A_r = \{x \in \Z_p \ | \ x \equiv r \pmod{2^l} \}$ to be the class in $\Z_p$ for which the remainder mod $2^l$ is $r$. We define
\[ h_{A}(x) = h_{z+H}(x) = \begin{cases} \frac{p}{2^l} \chi_{z}(x), & \text{if }x\in H^\perp, \\ 0, &\text{otherwise.}\end{cases} \]
It turns out that this function, which is a simple adaptation of (\ref{Zp_filter}) to $\Z_p$, is a `noisy' version of a `pure' filter function: the size of the coefficients $|\widehat{h_{A}}(\alpha)|$ is close to $1$ for $\alpha \in A$ and close to $0$ for $\alpha \notin A$. Indeed,
\[ \widehat{h_A}(\alpha) = \frac{1}{p} \sum_{x \in \Z_p} h_A(x) \overline{\chi_{\alpha}}(x) = \frac{1}{2^l} \sum_{x \in H^\perp} \chi_{z-\alpha}(x) \, .\]
Write $\alpha = 2^lk+j$, $z = 2^lq+r$ and $x = d 2^{-l}$ for $0 \leq j,r < 2^l$ and $0 \leq d \leq \lfloor \frac{p}{2^l} \rfloor$. Then,
\[ \widehat{h_A}(\alpha) = \frac{1}{2^l} \sum_{0 \leq d \leq \lfloor \frac{p}{2^l} \rfloor} \chi_{2^lq+r-2^lk-j}(d 2^{-l}) = \frac{1}{2^l} \sum_{0 \leq d \leq \lfloor \frac{p}{2^l} \rfloor} \chi_{q-k}(d) \chi_{r-j}(2^{-l}d) \, .\]
One can show that the last sum is large if and only if $j=r$ as $\chi_{r-j} \equiv 1$, that is if and only if $\alpha \in A$, and so that $|\widehat{h_A}(\alpha)| \approx 1$, and otherwise it is close to $0$. More precisely, for $\alpha = z$ we have $|\widehat{h_A}(\alpha)| = 1$ and as $k$ gets further away from $q$, the size of $\widehat{h_A}(2^lk+r)$ slowly decays (follows from Claim \ref{fact1} below). The function $h_A$ is said to be ``centered around" $z$. The results in Section~\ref{sec:mod-switch} below give further insights for the reasons why this adaptation of the filter function from $\Z_{2^n}$ to $\Z_p$ in the time domain, only slightly affects its frequency domain.

\noindent \textbf{The work of AGS.} The approach taken in \cite{AGS,AKA-thesis} is to work over intervals. We show how, using the scaling and frequency-shifting properties, one can transform from the set $A$ to an interval $I$ of the same size. Define $h_I(x) := h_A(2^{-l} x)$, then $\widehat{h_I}(\alpha) = \widehat{h_A}(2^l \alpha)$. This is a permutation of the coefficients of $h_A$. If $A = \{ r, 2^l+r, \ldots, t2^l+r\}$, then $I = \{r2^{-l},r2^{-l}+1,\ldots,r2^{-l}+t\}$, and the coefficients which were large on $A$ and small outside $A$ are now large over $I$ and small outside it. Moreover, if we define $h_I(x) := h_A(2^{-l} x) \chi_c(x)$ then by the shifting property the previous interval $I$ shifts to $I-c$.

AGS consider an interval $[a,b]$ of size $\lfloor \frac{p}{2^l} \rfloor$, for which $c = \lfloor \frac{a+b}{2} \rfloor$ is a middle point. They then define
\[h_{a,b}(x) = \begin{cases} \frac{p}{2^l}\chi_{c}(x), & \text{if }0\leq x < 2^l,\\ 0, & \text{otherwise.}\end{cases}\]
A direct calculation using the definition of $\widehat{h}_{a,b}(\alpha)$ shows that
\[\widehat{h}_{a,b}(\alpha) = \E_{0\leq x < 2^l}\left[\chi_c(x)\overline{\chi_\alpha(x)}\right] = \E_{0\leq x < 2^l}\left[ \chi_{c-\alpha}(x)\right].\]
Again, one can show that $|\widehat{h}_{a,b}(\alpha)| \approx 1$ if $a\leq \alpha \leq b$ and $|\widehat{h}_{a,b}(\alpha)| \approx 0$ for $\alpha$ outside this interval (see, for example, Claim \ref{fact1}). For further details see \cite{AGS,AKA-thesis}. This function is ``centered around" $c$, that is, for $\alpha = c$ we have $|\widehat{h_A}(\alpha)| = 1$ and while $\alpha$ gets further away from $c$, the size of $\widehat{h_A}(\alpha)$ slowly decays.

\begin{remark}\label{rem_intervals}
    There is a technical issue which we ignore in this description. As the size of $\widehat{h_A}(\alpha)$ slowly decays while $\alpha$ moves away from $c$, when $\alpha$ reaches the end of the interval $[a,b]$ the value $|\widehat{h_A}(\alpha)|$ is close to the value $|\widehat{h_A}(\beta)|$ for $\beta$ just outside this interval. This imposes some complexities in the filtering process; specifically one should take overlapping intervals, so the sets $A,B$ are not distinct as in the case of domains of size $2^n$. Moreover, the choice of the point $c$ (therefore the choice of the interval) also affects the filtering process. We refer to Sections $7.2.3$ and $7.2.4$ in \cite{AGS} and to \cite[Section 3]{AKA-thesis} for the technical details.
\end{remark}

With this filter function (either $h_A$ or $h_{a,b}$) $f_A$ can be approximated efficiently, as shown in the previous section. The algorithm now proceeds as the algorithm for domains of size $2^n$.

\subsection{Working with unreliable oracles}
\label{sec:noisy-oracle}

It is sometimes desirable to describe access to the function $f$ as querying an oracle. The oracle can be perfect -- always provides the correct value $f(x)$ -- or imperfect. Working with unreliable oracles is of importance in several applications. This section is dedicated to analyzing these cases.

Sometimes the samples $f(x_i)$ are given by an unreliable oracle $O$. By this we mean the oracle satisfies $O(x) = f(x)$ only with high probability. One can think of $O$ as a `noisy version' of $f$.
A common approach to this situation is to generate several independent values, each of which gives the value $f(x)$ with good probability; then, by applying majority rule, one can obtain the correct value $f(x)$ with overwhelming probability. Examples of this approach are presented in Section~\ref{sec:history}.

We show how the language of Fourier analysis gives a very general approach to analyze situations for working with unreliable oracles. The main idea is that if a function $f$ has a significant Fourier coefficient, then its noisy version also has a significant coefficient.
%(DO WE WANT TO ADD THIS: on the other hand,
Note however that if $f$ is concentrated, then its `noisy' version is not necessarily concentrated.

To be precise, let $f : G \to \C$. We describe the oracle as a function $O : G \to \C$ such that $O(x) = f(x)$ on the majority of $x \in G$. We assume that $\| O \|_\infty \leq \| f \|_\infty$. Define $R : G \to \C$ by $R(x) = O(x) - f(x)$ and let $I = \{ x \in G : R(x) \ne 0 \}$. We want to show that if $\widehat{f}(\alpha)$ is $\tau$-heavy, then $\widehat{O}(\alpha)$ is $\tau'$-heavy, for some $\tau'$ relatively large (its precise size depends on the success rate of the oracle).

Since $O=f+R$, then $\widehat{O}(\alpha) = \widehat{f}(\alpha) + \widehat{R}(\alpha)$.
% Note than $R$ is nonzero on a small set, denoted by $I$, and therefore has a small impact on the value $\widehat{O}(\alpha)$. Formally,
Note that $\|R\|_\infty \leq 2\|f\|_\infty$.
Hence
\[ \left|\widehat{O}(\alpha)\right| \geq \left|\widehat{f}(\alpha)\right| - \left|\frac{1}{|G|} \sum_{x \in I}R(x) \overline{\chi_\alpha}(x)\right| \geq \left|\widehat{f}(\alpha)\right| - \frac{2|I|}{|G|} \|f\|_\infty \,.\]
As $I$ is small, if $\widehat{f}(\alpha)$ is significant then so is $\widehat{O}(\alpha)$. Note that as the reliability rate of the oracle decreases, so does the size of $\widehat{O}(\alpha)$, while other coefficients increase in size. One can see that, similarly to majority rule, more samples are needed when the reliability rate of the oracle decreases. Indeed, the number of samples is proportional to $\tau^{-1}$ and as the size of the threshold $\tau$ decreases, $\tau^{-1}$ increases.

It is well-known that the GL theorem finds the unique function in case of low noise rate, namely if the the noise rate is smaller than $\frac{1}{4}-\epsilon$. One immediately sees this from our analysis:
the original function satisfies $|\widehat{f}(s)|=1$, for the secret vector $s$, and so only one Fourier coefficient of $O$ is larger than $\frac{1}{2}$.

\subsection{Hardness of finding significant coefficients in the random access model}
\label{sec:hardness-LPN}

The SFT algorithm requires chosen queries. The aim of this section is to explain that one does not expect a general learning algorithm for problems where the function values cannot be chosen.
Indeed, we will show that if such a learning algorithm existed then the learning parity with noise (LPN) and learning with errors (LWE) problems would be easy.

Recall the LPN problem: an instance is a list of samples $(a, b = \langle a,s \rangle + e(a)) \in \Z_2^n \times \Z_2$ for some secret value $s$ and a function $e :\{0,1\}^n \to \{0,1\}$ which determines the noise.
Define $\text{LPN} :\{0,1\}^n \to \{0,1\}$  by $\text{LPN}(a) := (-1)^b$. This is a `noisy version' of the function $f(x):=(-1)^{\langle a,s \rangle}$ for which $\widehat{f}(s)$ is the only non-zero Fourier coefficient. For a small noise rate (as in LPN), as shown in Section~\ref{sec:noisy-oracle}, the coefficient $\widehat{\text{LPN}}(s)$ is a significant coefficient for this function. Hence, if one could find significant coefficients in $\{0,1\}^n$ on random samples, then one could solve LPN given the samples $(a,b)$. Since LPN is believed to be hard, one does not expect such a variant of the SFT algorithm to exist. Further evidence for the hardness of this problem in the random access model is that it is related to the problem of decoding a random binary linear code.

The same argument holds for LWE in $\Z^n_p$. In LWE one has samples $a \in \Z_p^n$ and $b = \langle a, s \rangle + e(a) \pmod{p}$ where $e(a)$ is ``small'' relative to $p$. Defining $\text{LWE}(a) := \wp^b$ one can show that the coefficient of the character $\chi_s(x) = \wp^{\langle x, s \rangle}$ is significant. Hence, if one could find the significant coefficients when given random samples, then one could solve LWE given the samples $(a,b)$. Since we  have good evidence that LWE is a hard problem, this shows that we do not expect to be able to learn significant Fourier coefficients in the random access model.

The modulus-dimension tradeoff for LWE \cite{mod-dim} shows how to transform LWE in $\Z^n_p$ to LWE in $\Z^{n/d}_{p^d}$ (albeit with a different error distribution), and so one can conclude that finding significant coefficients in $\Z_{p^n}$ on random samples is at least as hard as solving LWE in $\Z^n_p$ with binary secrets.
This is an example of the connection between $\Z_2^n$ and $\Z_{2^n}$ as explained in Remark~\ref{rem_mod-dim}.

\section{Simplifications: Modulus Switching} \label{sec:mod-switch}

The SFT algorithm is considerably simpler to understand and implement for $\Z_2^n$ or $\Z_{2^n}$ than for $\Z_p$.
Furthermore, for domains of size $2^n$, considerable effort has been invested by researchers in the engineering community into making this algorithm more efficient with respect to various measures~\cite{Fourier_survey} (see also Mansour and Sahar~\cite{ManSah}).
Hence, it is natural to try to work with functions over $\Z_{2^n}$ instead of functions over $\Z_p$.
We now sketch an approach that shows how one can transform functions on $\Z_p$ into functions on $\Z_{2^n}$ where $2^n \approx p$, while maintaining a relation between their significant coefficients.
In analogy to similar ideas in lattice cryptography we call this ``modulus switching''.

These ideas are implicit in the work of Shor~\cite{Shor} on factoring with quantum computers.
Shor extends a periodic function to a larger domain. The core idea is that if a function is periodic, then the period, which is a feature of the time domain, is preserved over any (large enough) domain. This fact is exploited by Shor, where his further ideas take place in the frequency domain.
Shor's analysis provides a clear interaction between the representation of a (periodic) function in the time and frequency domains.

We extend these ideas to show that a much larger class of functions keeps the properties of their frequency domain representation, when extending their time domain. Specifically, significant coefficients are ``preserved'' even when the time domain representation of the function is extended (by ``preserved'' we mean that there is a clear relation between the significant coefficients of both functions).
We refer to Laity and Shani~\cite{LS} for the technical details.

Let $N = 2^n > p$ be the smallest power of two greater than $p$. For a function $f:\Z_p\to\C$, we define
\[
    \widetilde{f}(x) := \left\{ \begin{array}{ll} f(x) & \text{ when } 0 \le x < p , \\
        0         & \text{ when } p \le x < N .\end{array} \right.
\]
Note that the operation $f \mapsto \widetilde{f}$ is $\C$-linear.
The basic observation (see Figure~\ref{characterMS}) is that for a character $\chi_{\alpha}$ on $\Z_p$,
%\[
%  \widetilde{\chi}_\alpha(x) := \left\{ \begin{array}{ll} \chi_{\alpha}(x) & \text{ when } 0 \le x < p , \\
%     0 & \text{ when } p \le x < N \end{array} \right.
%\]
$ \widetilde{\chi}_\alpha(x)$ is a function on $\Z_N$ that is also concentrated.

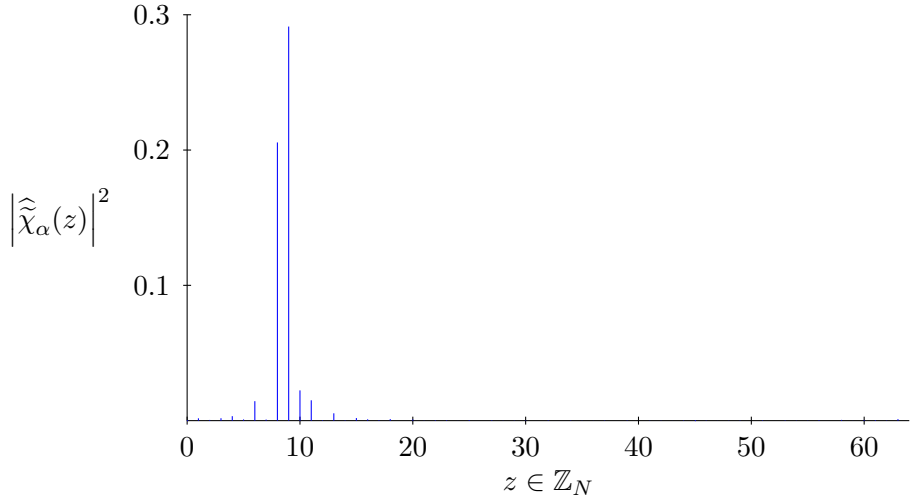
\begin{figure}[!h]
    \begin{center}
        \begin{tikzpicture}[scale = 1.5]
            % on yaxis 12cm = 1 unit
            \def \yscale{12}
            % on x axis 0.1cm  = 1 unit
            \def \xscale{0.1}

            \draw (0,0) -- (0,\yscale*0.3) node[pos=0.5,left=25pt] {$\left|\widehat{\widetilde{\chi}}_\alpha(z)\right|^2$};
            \draw(0,0) -- (64*\xscale ,0) node[pos=0.5,below = 15pt] {$z\in \mathbb{Z}_N$};

            \foreach \y in {0.1,0.2,0.3}
            \draw (-1pt, \y*\yscale cm) -- (1pt , \y*\yscale cm) node[left = 4pt] {$\y$};

            \foreach \x in {0,10,...,63}
            \draw (\x*\xscale cm , -1pt) -- (\x*\xscale cm, 1pt) node[below=4pt] {$\x$};

            \draw[color = blue] plot[ycomb] file{charMS.txt};
        \end{tikzpicture}
    \end{center}
    \caption{The magnitude of the Fourier coefficients  $\widehat{\widetilde{\chi}}_\alpha(z)$. Here  $p=37$, $N=64$ and $\alpha = 5$.}\label{characterMS}
\end{figure}
To explain this observation we state the following basic fact and sketch a proof of it.
It is straightforward to turn this result into a rigorous upper bound.

\begin{claim} \label{fact1}
    Let $N > 1$, $\omega_N = \mathrm{e}^{\frac{2 \pi i}{N}}$ and let $\alpha \in \R$, $\alpha \ne 0$, $|\alpha| < N/2$ and $K \in \N$. Define
    \[
        S_{\alpha, K} = \sum_{x = 0}^{K-1} \omega_N^{\alpha x } \,.
    \]
    Then
    \[
        | S_{\alpha, K} | \approx N \frac{|1 - \omega_N^{\alpha K} |}{2 \pi |\alpha| } \,.
    \]
\end{claim}

To see this note that the geometric series  sums to $(1 - \omega_N^{\alpha K} ) / ( 1 - \omega_N^{\alpha} )$ and the denominator is $   ( 1-  \cos( 2 \pi \alpha/N)) - i \sin( 2 \pi \alpha / N )$
which has norm squared equal to $2(1 - \cos( 2 \pi \alpha / N ) )$.
Finally, since $(1  - \cos(x)) \approx x^2/2$ (indeed $\tfrac{x^2}{2}( 1 - \tfrac{x^2}{12}) \le 1 - \cos(x) \le \tfrac{x^2}{2}$), the result follows.

\bigskip
We now compute the Fourier transform of $\widetilde{\chi}_\alpha$ as a function on $\Z_N$ where $N = 2^n$.
We have
\[
    \widehat{\widetilde{\chi}}_\alpha ( \beta ) = \langle \widetilde{\chi}_\alpha, \chi_\beta \rangle = \frac{1}{N} \sum_{x=0}^{p-1} \exp\left( 2 \pi i \left( \tfrac{\alpha}{p} - \tfrac{\beta}{N} \right) x \right ).
\]
If $\tfrac{\alpha}{p} - \tfrac{\beta}{N} \ne 0$, which will be satisfied in general since $\alpha, \beta  \in \Z$ while $\gcd( p, N) = 1$, then applying Claim~\ref{fact1} gives the approximation
\[
    \left|\widehat{\widetilde{\chi}}_\alpha ( \beta ) \right|\approx \frac{|1 - \exp( 2 \pi i( \alpha/p - \beta/N))|}{2 \pi |\alpha/p - \beta/N|} \,.
\]
If $\beta \approx N \alpha / p$ then this coefficient is large and so the function $\widetilde{\chi}_\alpha$ has a significant Fourier coefficient at $\lfloor N \alpha / p \rceil$. Moreover, the size of $\widehat{\widetilde{\chi}}_\alpha(\lfloor N \alpha / p + k \rceil)$, for $0<|k|<N/2$, is bounded by $O({1}/{k})$, and so $\widetilde{\chi}_\alpha$ is concentrated in a small set $\Gamma \subseteq \Z_N$ of characters represented by values around $N \alpha / p$.

Since the maps $f \mapsto \widetilde{f}$ and $g \mapsto \widehat{g}$ are $\C$-linear, for any $f(x) = \sum_{\alpha\in G} \widehat{f}(\alpha)\chi_\alpha(x)$ we have
\[
    \widehat{\widetilde{f_{\;}}}(\beta)= \sum_{\alpha=0}^{p-1} \widehat{f}(\alpha) \widehat{\widetilde{\chi}}_\alpha ( \beta ) \,.
\]
Thus, if $\widehat{f}(\alpha)$ is a significant coefficient for $f$, then one expects that for $\beta = \lfloor N \alpha / p \rceil$, the coefficient $\widehat{\widetilde{f_{\;}}}(\beta)$ is significant for $\widetilde{f}$.
The work of Laity and Shani~\cite{LS} made these arguments to a precise theorem.

\begin{theorem}[{\cite[Theorem 1.1]{LS}}]\label{thm:LS}
    Let $\{n_k\}_{k\in\N},\{m_k\}_{k\in\N}$ two sequences of positive integers with $m_k \geq n_k/2$ for every $k\in\N$.
    Let $Q\in \R[x]$ be a polynomial.
    Let $\{f_k:\Z_{n_k}\to\C\}_{k\in\N}$ be a concentrated family of functions such that $\|f_k\|_2^2 \leq Q(\log(n_k))$ for all $k\in \N$.
    Then $\{\widetilde{f_k}:\Z_{m_k}\to\C\}_{k\in\N}$ is a concentrated family of functions.
\end{theorem}

Specifically, if $f(x)$ is a concentrated function on $\Z_p$ then $\widetilde{f}(x)$ is a concentrated function on $\Z_{2^n}$. A similar result holds where $f$ is $\epsilon$-concentrated. We refer to~\cite{LS} for the technical details.

As a consequence, one sees that it is not necessary to develop a variant of the SFT algorithm for the group $\Z_p$. Instead one can simply modulus-switch to a power of two and apply the SFT algorithm for the group $\Z_{2^n}$. This is addressed in \cite[Section 6.1]{LS}.
Since the algorithms for $\Z_{2^n}$ have been optimised significantly (see~\cite{Fourier_survey,ManSah}) we believe that the resulting algorithms will be no less efficient than applying the AGS algorithm directly. Moreover, unlike the complexities working directly over $\Z_p$ as explained in Remark \ref{rem_intervals}, this technique (although it might introduce new ``noise'') overcomes the need to take overlapping intervals and is not subject to the choice of the interval.

%The above discussion assumed the new function $\widetilde{f}$ extend $f$ from $\Z_p$ to $\Z_{2^n}$ where $2^n$ is slightly larger than $p$. As Theorem \ref{thm:LS} shows, one can consider modulus switching for domains of any size, including switching to a smaller domain.
%The above discussion applies to extending a function from $\Z_p$ to $\Z_{2^n}$ where $2^n$ is slightly larger than $p$, and this was done by setting the new function to be zero on the new values.
%There are other ways to extend a function, for example using periodicity or some ``intrinsic'' description of the function (such as in the case of $i$-th bit functions).
%More generally one can consider modulus switching for domains of any size, including switching to a smaller domain.
%The results about concentration hold in this greater generality, and this provides a new technique to prove concentration of (some) families of functions, by showing that a subfamily, for domains of specific forms, is concentrated.
%We see an example of this in the next subsection.

\subsection{The $i$-bit function is concentrated} \label{sec:i-th-bit}

We now explain that modulus switching provides an alternative proof of the Morillo--R{\` a}fols result that every single-bit functions is concentrated~\cite{MR}.

The above discussion assumed the function $\widetilde{f}$ extends $f$ from $\Z_p$ to $\Z_{2^n}$ where $2^n$ is slightly larger than $p$. As Theorem \ref{thm:LS} shows, one can consider modulus switching for domains of any size, including switching to a smaller domain.
The results about concentration hold in this greater generality, and this provides a new technique to prove concentration of (some) families of functions, by showing that a subfamily of functions, defined on domains of specific forms, is concentrated.

\begin{theorem}[{\cite[Theorem 6.1]{LS}}]\label{thm:mod-switch}
    Consider a family of functions $\famly = \{ f_{2^k}:\Z_{2^k} \to \C \}_{k\in\N}$ and define the family $\famly' = \{ f_n:\Z_n \to \C \}_{n\in\N}$, where for each $2^{k-1} < n \leq 2^k$ we let $f_n(x) := f_{2^k}(x)$ for every $x\in\Z_n$.
    If $\famly$ is concentrated then $\famly'$ is concentrated.
\end{theorem}

As an application, one can prove that the $i$-th bit function is concentrated by showing that the family of the $i$-th bit function on domains $\Z_{2^k}$ is concentrated, that is, that $\{\bit_i : \Z_{2^k} \to \{-1,1\}\}_{i<k\in\N}$ is concentrated. Here $i$ can be a function of $k$, so for example the most-significant-bit function is given by $i=k-1$. The latter can be easily proven using the structure of these functions under these domains. This is summarized in the following lemma, where we define $|x|_N := \min \{x,N-x\}$.

\begin{lemma}[{\cite[Lemma 6.2]{LS}}]
    Let $k \in \N$ and $0 \le i < k$.
    Define $\bit_i : \Z_{2^k} \to \{-1,1\}$ by
    $\bit_i(x) = (-1)^{x_i}$ where $x = \sum_{j=0}^{k-1} x_j 2^j$ and $x_j \in \{ 0,1 \}$.
    Let $\alpha \in \Z_{2^k}$.
    Then $\widehat{\bit_i}(\alpha) = 0$ unless $\alpha$ is an odd multiple of $2^{k-i-1}$ in which case $| \widehat{\bit_i}(\alpha)| = O( 2^{k-i}/|\alpha|_{2^k} )$.
\end{lemma}

The lemma shows that, when $i$ is small there are a few non-zero coefficients (especially for $i=0$, there is only one non-zero coefficient at $\alpha = 2^{k-1}$).
When $i$ is ``medium'' then there are non-zero coefficients at all multiples $\alpha = j 2^{k-i-1}$, $j$ odd, and they decrease in size with $1/|j|_{2^k}$.
When $i$ is large (e.g., $i = k-1$) then the significant coefficients are all close to $0$ and are spaced at distance $2\cdot2^{k-1-i}$ (i.e., when $i = k-1$ they are 2 apart; for the second most significant bit they are spaced 4 apart, and so on).

A corollary is that the $i$-th bit function on $\Z_{2^k}$ is concentrated. See arguments on the function $\half$ in Example~\ref{example4} and~\cite[Claim 4.1]{AGS}. For clarification, we state again that $i$ can be a fixed constant ($i=0$ corresponds to the least significant bit) or be dependent on $k$ ($i=k-1$ corresponds to the most significant bit).

Having established that the $i$-th bit function is concentrated on $\Z_{2^k}$, our modulus switching approach shows that the $i$-th bit function (on any domains $\Z_N$) is concentrated by Theorem~\ref{thm:mod-switch}. This general approach gives a new and simpler proof of the result in~\cite{MR} (the proof in~\cite{MR} is very technical; they decompose $N = k 2^i \pm m$ and consider different cases of $m$).

\section{Applications: Cryptography}
\label{sec_applications}

The SFT algorithm has been used to reprove known results on the hardness of recovering bits of the secret values in the discrete logarithm problem (DLP) and RSA problem. It has been used to give reductions for the learning with errors (LWE)~\cite{LWE} and learning with rounding (LWR)~\cite{LWR} problems, that prove that the `search' and `decision' problems are equivalently hard even when the number of samples is fixed. It has also been used to prove results about the hardness of recovering bits of Diffie--Hellman shared secrets keys in both (non-prime) finite fields and elliptic curves. This section surveys how the SFT algorithm is used in these applications. In addition, we explain the specific model for which the Diffie--Hellman results hold, and clarify that the question whether single bits of Diffie--Hellman shared keys are hardcore (in the usual model) is still open.

%We emphasize that these results are only partial and very restricted. The main question of whether one can learn some bits of $g^{ab}$ given $g,g^a,g^b$ is still open for both of these domains.

\subsection{Background and motivation}\label{sec:HNP}

A one-way function $h$, if it exists, assures that while given $x$ it is easy to compute $h(x)$, retrieving $x$ from $h(x)$ is hard. This hardness does not necessarily mean that given $h(x)$ one cannot find some partial information of $x$. Naturally, the main interest is in trying to learn some bits of $x$, but other sorts of partial information have also been considered. Bits of $x$ that cannot be learnt from $h(x)$, or more generally cannot be predicted noticeably better than a guess, are called \emph{hardcore bits}. In other words, a hardcore bit is a bit which is as hard to compute (or to predict) as the entire secret value.
For a historical overview see \cite{HardCoreSurvey}.
To show that a bit (or a set of bits) is hardcore, one usually tries to construct an algorithm that inverts $h$, given a target value $h(x)$ and an oracle that takes $h(t)$ and outputs a bit of $t$. In order to do so, one first needs to establish a way to query the oracle on values $h(t)$ such that there is some known relation between $t$ and $x$, for example $t = \alpha x$ for known $\alpha$'s.

A useful language to describe these ideas is the \emph{hidden number problem}, which was introduced by Boneh and Venkatesan \cite{BV} in order to study bit security of secrets keys arising from Diffie--Hellman key exchange. This problem turned out to be general enough to be applied to other cryptographic problems like DLP and RSA. In fact, the generality of the problem allows it to be used also outside of the scope of bit security (see \cite[Section 4.4]{2faces} and references within, also \cite{De-Mulder, GLV/GLS}). Therefore, the hidden number problem is of theoretical interest and is studied today in its own right. It has many extensions and different variants; see \cite{Shparlinski} for a comprehensive survey.

%We define the hidden number problem (HNP) as follows.
\begin{definition} [Hidden number problem]
    \label{HNP_def}
    Let $(G,\cdot)$ be a group, let $s \neq 0$ be a secret (unknown) element of $G$ and let $f$ be a function defined over $G$. Find $s$ using oracle access to the function $f_s (x) := f(s \cdot x)$.
\end{definition}

We use the term \emph{oracle access} as a general term for either of the following oracle models: in the \emph{random access} model the solver receives polynomial many samples $(x, f_s(x))$ where the values $x$ are drawn independently and uniformly at random from $G$; in the \emph{query access} model the solver can query the oracle on any input $x \subseteq G$ and receive the answer $(x, f_s(x))$. To emphasize the difference between these models, we refer to the hidden number problem in the latter model as chosen-multiplier hidden number problem (CM-HNP). This problem can also be divided into two models, namely \emph{adaptive access} where the solver has a continuous access to the oracle and can query it at any time of the recovery process, and \emph{non-adaptive access} where the solver is not allowed to query the oracle once the recovery process has started. Other types of access models could be also considered. For example, the original work on the hidden number problem~\cite{BV} considers an oracle for which on the query $x \in \Z_p$ replies with $(x, f(s g^x))$.

An interesting case is when the oracle is unreliable. That is, the oracle does not give a correct answer all the time, but with some probability. It is common to call an oracle that always provides a correct answer a \emph{perfect} oracle. An oracle that is correct only with some noticeable advantage is called an \emph{unreliable} or \emph{imperfect} oracle.

The following table summarizes some of the known results on the hidden number problem in different models. Here $p$ is a prime number and `imperfect' under the `Oracle' column refers to an oracle with any non-negligible advantage over trivial guessing.
The starting point of this work is the Boneh--Venkatesan result~\cite{BV} which requires a perfect oracle and uses lattice methods rather than Fourier learning methods; this work was adapted to unreliable oracles by \cite{GVNS}, but there is a complex  tradeoff with the number of bits and so we do not include it in our table.

\begin{center}
    \begin{tabular}{ | l | l | l | p{4.8cm} | l | l |}
        \hline
        Problem & Access       & Group    & Bits                                            & Oracle    & Remarks                        \\ \hline
        HNP     & random       & $\Z^*_p$ & $\sqrt{\log p} + \log\log p$ MSB\footnotemark   & perfect   & Given by \cite{BV}             \\ \hline
        CM-HNP  & adaptive     & $\Z^*_p$ & LSB                                             & imperfect & Given by \cite{ACGS}           \\ \hline
        CM-HNP  & adaptive     & $\Z^*_p$ & any single bit                                  & imperfect & Given by \cite{HN}             \\ \hline
        CM-HNP  & non-adaptive & $\Z^*_N$ & MSB \& LSB                                      & imperfect & Given by \cite{Bleichenbacher} \\ \hline
        CM-HNP  & non-adaptive & $\Z^*_N$ & each single bit for the outer $\log\log p$ bits & imperfect & Given by \cite{AGS}            \\ \hline
        CM-HNP  & non-adaptive & $\Z^*_N$ & any single bit                                  & imperfect & Given by \cite{MR}             \\ \hline
    \end{tabular}
\end{center}
\footnotetext{Since one can easily transform HNP with the LSB function to HNP with the MSB function, HNP can also be solved given $\sqrt{\log p} + \log\log p$ LSB. A generalization of this technique \cite[Section 5.1]{NguShpa} allows to transform HNP with $2d$ consecutive inner bits to HNP with $d$ MSB, hence HNP can also be solved given $2(\sqrt{\log p} + \log\log p)$ consecutive inner bits.}

Most early works such as~\cite{ACGS,Bleichenbacher,HN} require complicated algebraic manipulations such as tweaking and untweaking bits. Using the SFT algorithm~\cite{AGS} gives a uniform and clear approach. We present this solution to CM-HNP, using different terminology than the original one, for functions of norm $1$, as the subsequent applications involve single bit functions (with the convention that $\bit_i(x) = (-1)^{x_i}$ where $x_i$ is the $i$-th bit of $x$).

\begin{theorem} [\cite{AGS}]
    \label{HNP_thm}
    Let $f: \Z_N \to \{-1,1\}$ be a function with a $\tau$-heavy Fourier coefficient $\alpha\in\Z^*_N$ for $\tau^{-1}=poly(\log|G|)$. Then, the chosen-multiplier hidden number problem in $\Z_N$ with $s \in \Z^*_N$ and the function $f$ can be solved in polynomial time.
\end{theorem}
In particular, the theorem holds for every concentrated function.

\begin{remark}[Coding Theory terminology]
    Theorem \ref{HNP_thm} rephrases Theorem $2$ of \cite{AGS}. The latter work gives a polynomial time list-decoding algorithm for \emph{concentrated codes} with corrupted code words (Theorem $1$) and subsequently a general list-decoding methodology for proving hardcore functions (Theorem $2$). Most subsequent works on hardcore bits adopt this coding-theoretic language.
    Thus, in order to apply Theorem $2$ of \cite{AGS}, these works use Theorem $1$ of \cite{AGS}, which applies to concentrated codes. This caused the authors of these works to put effort into proving that a particular code is concentrated. However, we emphasize that to apply the CM-HNP approach of~\cite{AGS} there is no need for the function to be concentrated.
    Instead it suffices that the function has a significant Fourier coefficient, and this is usually much easier to prove.
    We make this clear in our formulation of Theorem~\ref{HNP_thm}.
    In other words, while concentration is sufficient for a code to be \emph{recoverable} it is not a necessary condition.
    For these reasons (and others) we find the coding-theoretic language unhelpful and do not use it in this paper.
\end{remark}

We now sketch the proof of Theorem~\ref{HNP_thm}: run the SFT algorithm on $f$ and $f_s$ to get short lists $L,L_s$ of $\tau$-heavy coefficients for each function, respectively. By the scaling property $\widehat{f_s}(\alpha) = \widehat{f}(\alpha s^{-1})$ for every $\alpha$. Therefore, for every $\alpha \in L_s$ for which $\widehat{f_s}(\alpha)$ is $\tau$-heavy there exists $\beta \in L$ such that $\beta = \alpha s^{-1}$. The secret $s$ can be recovered efficiently. Notice that while the hidden number problem takes place in a multiplicative group, this solution involves Fourier analysis over an additive group.

%Since the algorithm in the previous section learns heavy Fourier coefficients in the query access model, we get the following result.
%
%\begin{corollary}[\cite{AGS}]
%\label{cor_HNP-CM}
%For any concentrated function $f: \Z_N \rightarrow \{-1,1\}$, there exists an algorithm that solves CM-HNP in $\Z^*_N$ in polynomial time, where the queries are made non-adaptively.
%\end{corollary}

%Therefore, the study of CM-HNP in $\Z^*_N$ can be reduced to the study of concentrated functions. Any single bit function is concentrated, as shown by \cite{MR} \textbf{[and in section XX above?]}. Therefore, non-adaptive CM-HNP in $\Z^*_N$ can be solved for any single bit function.

\medskip

A template for algorithms for CM-HNP is the following:
%This allows the solvers of CM-HNP to change their viewpoint, and instead of trying to solve the problem directly, focusing on showing
show that $(i)$ the ``partial information'' function $f$ has a significant coefficient, $(ii)$ the function $f_s$ has a significant coefficient, and $(iii)$ some (recoverable) relation between the coefficients of $f$ and $f_s$ exists. If one succeeds in showing these conditions, then using the SFT algorithm one can solve this instance of CM-HNP.
This template allows bit security researchers to look for settings where a solution to CM-HNP is already known (namely, cases where these three conditions are already known to hold, like single-bit functions over $\Z_N$) and try to convert their problem of interest to this setting.
%These properties are often called \emph{codes with multiplicative access} in the literature, but we prefer not to use this language. \textbf{(I don't see the need to write it here. Either cut-off or move up)}

%In the following we show some applications using the latter approach. In Section \ref{sec_limitations} below, we show that the former approach is very limited. Namely, we show that if condition $(i)$ holds, then condition $(ii)$ can hold only if the function $f_s$ is an affine change of variable with respect to the function $f$. As a consequence, one cannot apply this solution to variants of the hidden number problem where more complex operations take place, like the case of elliptic curves.

\subsubsection{The multivariate hidden number problem}

Another case of interest is the \emph{multivariate hidden number problem} (MVHNP), which we define as follows.
\begin{definition} [Multivariate hidden number problem]
    \label{MVHNP_def}
    Let $R$ be a ring, let $\textbf{s}=(s_1,\dots,s_m)\ne (0,\dots,0)$ be a secret (unknown) element in $R^m$ and let $f$ be a function defined over $R$. Find $\textbf{s}$ using oracle access to the function $f_\textbf{s}(\textbf{x}) := f (\textbf{s} \cdot \textbf{x}) = f (s_1 x_1 + \cdots + s_m x_m$).
\end{definition}

Specific instances of this problem are LWE and LWR, and it is related to \emph{trace-HNP}~\cite{trace-HNP} and \emph{polynomial HNP}~\cite{Poly-HNP}. Similar to the solution to HNP in $\Z_p$, one can give a solution in $\Z_{p^m}$ in the random access model for a function $f$ that outputs $\sqrt{\log(p^m)} = \sqrt{m\log(p)}$ MSB's of its input (derived from~\cite{Poly-HNP}, for example). %Furthermore, when $p$ is large compared to $m$, it is possible to adjust this solution even when only $\sqrt{\log(p)}$ MSB's are given. The resulting algorithm is a rather straightforward lattice attack, very similar to known algorithms for LWE, that reduces the problem to the closest vector problem in a certain lattice.

One can also define CM-MVHNP, the chosen-multiplier version of the multivariate hidden number problem, similar to CM-HNP. To solve this variant we need an analogue of the Fourier scaling property in higher dimensions. Such an analogue, which we call the multivariate scaling property, is given in~\cite[Lemma 13]{MVHNP} and we sketch it now.

\noindent\emph{Multivariate scaling property.} Let $f: \Z_p \rightarrow \C$, let $\textbf{s}=(s_1,\dots,s_m) \in \Z_p^m$ such that not all $s_i = 0$, and define $f_\textbf{s}: \Z_p^m \rightarrow \C$ by $f_\textbf{s}(\textbf{x}) := f (\textbf{s} \cdot \textbf{x})$. For any $s_k \neq 0$, the Fourier transform of $f_\textbf{s}$ satisfies
\begin{align*}
    \label{eq:MVscaling}
    \widehat{f_\textbf{s}}\left(\mathbf{z}\right) = \widehat{f_\textbf{s}}(z_1,\dots,z_m) = \left\{
    \begin{array}{l l}
        \widehat{f}(c) & \quad \mathrm{if} \ \ (z_1,\dots,z_m) = (cs_1,\dots,cs_m), \ \  c \in \Z_p \,; \\
        0              & \quad \mathrm{otherwise.}
    \end{array}
    \right.
\end{align*}

This allows generalizing Theorem \ref{HNP_thm} to CM-MVHNP. The proof, which we omit, follows from the proof to Theorem~\ref{HNP_thm} given above.

\begin{theorem} [\cite{MVHNP}]
    \label{MVHNP_thm}
    Let $f: \Z_p \to \{-1,1\}$ be a function with a $\tau$-heavy Fourier coefficient $\alpha\in\Z^*_p$ for $\tau^{-1}=poly(\log|G|)$. Then, the chosen-multiplier multivariate hidden number problem in $\Z_p^m$ with the function $f$ can be solved in polynomial time.
\end{theorem}

\subsection{Applications}\label{sec:applications}

We present some of the applications in cryptography of the SFT algorithm. They are all based on reducing some problems to the CM-HNP or CM-MVHNP. In the following we assume to have an oracle $O$ that solves some problem, and show how to use this oracle to solve a harder problem, thus establishing the hardness equivalence between the two problems.

%Let $f$ be a one-way function and let $O$ be an oracle that on $f(x)$ returns a single bit of $x$. One proves that this bit is hard to calculate as inverting $f$ by showing a reduction from inverting $f$ to the oracle $O$. That is, one shows that with a few queries to the oracle one can learn $x$, thus inverting $f$. In order to do so, one first needs to establish a way to query the oracle on values $f(t)$ such that there is some known relation between $t$ and $x$. The hidden number problem is where $t = \alpha x$ for known $\alpha$'s.

\subsubsection{Proving known results: bit security of RSA and DLP}

The first application of the algorithm was given in~\cite{AGS}, where it is shown that the most significant bit and least significant bit are hardcore for the RSA function $RSA_{N,e}(x) := x^e \pmod{N}$ and for exponentiation $EXP_{g}(x) := g^x$, where $g$ is an element of prime order $\ell$ in some group. The results hold for imperfect oracles that have noticeable advantage over guessing. These results were already known, as \cite{ACGS} first shows that the LSB is hardcore for the RSA function and \cite{HN} shows that every bit is hardcore for both functions. Nevertheless, the approach based on SFT is more general (holds for every function with significant coefficients) and simpler. We explain how to derive these results.

\begin{claim}
    Each single bit is hardcore for the RSA function. That is, predicting any bit of $x$, given $RSA_{N,e}(x)$, is as hard as inverting the RSA function.
\end{claim}

We sketch the proof: One direction is trivial. In the other direction, given an instance $RSA_{N,e}(x) = x^e \pmod{N}$, we want to recover $x$. Suppose the (imperfect) oracle $O$ takes $RSA_{N,e}(t) := t^e \pmod{N}$ and outputs $\bit_i(t)$, the $i$-th bit of $t$. Since the values $e,N$ are public in the RSA setting, for every number $r$ one can compute $RSA_{N,e}(rx \pmod{N})$ by $(r^e  \pmod{N})(x^e \pmod{N}) = (rx)^e \pmod{N}$. Hence, given $RSA_{N,e}(x)$ one can query the oracle on $RSA_{N,e}(rx)$ to get the $i$-th bit of $rx$ for every chosen $r$.
The problem therefore becomes the CM-HNP in $\Z_N^*$, and this can be solved using the SFT algorithm over the additive group $(\Z_N, + )$, which has known order. Indeed, $\bit_i$ is concentrated (see Section~\ref{sec:i-th-bit}), thus has a significant coefficient. The oracle function $O$ also has a significant coefficient (see Section~\ref{sec:noisy-oracle}). The rest follows from Theorem~\ref{HNP_thm}.

\begin{claim}
    Each single bit is hardcore for the exponentiation function $EXP_g$ for prime-order element $g$. That is, predicting any bit of $x$, given $EXP_{g}(x)$ is as hard as inverting the function $EXP_g$, i.e. solving DLP in the corresponding group.
\end{claim}

The proof, which we leave as an exercise, is similar to the previous case, using the fact that $(g^x)^r = g^{rx}$.
%Similarly, given an instance $EXP_{g}(x) = g^x $, since $g$ and $\ell$ are public, one can calculate $EXP_{g}(rx \pmod{\ell})$ for every number $r$ by $(g^x)^r = g^{rx}$. Thus, given $EXP_{g}(x)$ one can query the oracle on $EXP_{g}(rx)$ to get a bit of $rx$ for every chosen $r$.
%The problem therefore becomes the CM-HNP in $\Z_{\ell}^*$, and this can be solved using the SFT algorithm over the additive group $(\Z_{\ell}, + )$.
This proves bit security results for the DLP in finite fields and elliptic curves.
%Applying Theorem~\ref{HNP_thm} we find that all bits for those functions are hardcore.
%We therefore have a way to reduce our goal to recover $x$ to CM-HNP with single-bit functions. Since the latter has a solution for imperfect oracles with non-negligible advantage over a trivial guess, it shows that
Similar results also hold for other functions (problems), as \emph{Rabin} (see~\cite[Chapter 7]{AKA-thesis}) and the Paillier trapdoor permutation (see \cite[Section 7]{MR}).

\subsubsection{Bit security of the Diffie--Hellman protocol and related schemes}

An open question is to prove that single bits of Diffie--Hellman keys are hardcore. Here we consider an oracle $O$ that on $g,g^a,g^b$ returns a single bit of the Diffie--Hellman key $s = g^{ab}$. To interact with the oracle, notice that given $g^b$ one can compute $g^{b+r} = g^b g^r$ for any $r$. One can then query the oracle $O$ with $g,g^a,g^{b+r}$ and receive a bit of $g^{a(b+r)} = g^{ab} g^{ar} = s t$.
This is how the hidden number problem was originally identified.
This interaction does not correspond to the CM-HNP, since choosing the multiplier $t = (g^a)^r$ (for the secret $s$) is equivalent to finding discrete logarithms for the base $g^a$ in $\Z^*_p$.

\noindent \textbf{Advice bits.} For related schemes where the exponent $a$ is fixed (unlike schemes using ephemeral exponents, as in Diffie-Hellman key exchange), Akavia~\cite{AKA-HNP} followed Boneh--Venkatesen~\cite{BV2} to get around this problem by assuming an ``advice'' that provides the discrete logarithms of the chosen multipliers $t$ to the base $g^a$, but this is not realistic in actual applications (see also our remark in Section \ref{sec_subgroups}).
There is currently no method known to prove the hardness of single bits of Diffie--Hellman keys in the usual model.
%within a fixed representation of the group and without unrealistic advice. \textbf{(this seems out of the blue; we only talk about rep. changing below)}

%Therefore, the result above cannot be used directly on the computational Diffie-Hellman problem in $\Z^*_p$ (in related schemes, where one can generate multipliers that only depend on $g$ \textendash \space or if $g^a$ is a long-term value \textendash \space one can consider the discrete logarithms as \emph{advice}; see also our remark in Section \ref{sec_subgroups}).

\noindent \textbf{New Diffie--Hellman model.} To overcome this problem, Boneh and Shparlinski~\cite{BS} suggested (in the context of elliptic curves) a different model where the oracle $O$ takes as input, in addition to the values $g, g^a, g^b$, a group homomorphism $\phi: G \to G'$, and then outputs partial information (e.g. a single bit) of $\phi(g^{ab})$. %The usual case corresponds to $\phi$ being the identity function.
The approach is then to keep the inputs $g,g^a,g^b$ fixed and to use $\phi$ as the way to choose multipliers for $s=g^{ab}$ in the hidden number problem. This model corresponds to a variant of Diffie--Hellman key exchange, where a representation of the group is not fixed.
% Notice that if one , then interaction with the oracle takes place only with the isomorphisms $\phi$.
We call this the \emph{representation changing model}.
%This model allows to convert the nonlinear Diffie--Hellman problem to an easier linear problem. BARAK: surely, the DHP is not even considered here!

This is an example of our discussion at the end of Section~\ref{sec:HNP} above on converting a given problem to a setting (a new model, in this case) that allows to apply the solution to CM-HNP. We now explain how in this model one can reduce the original problem to variants of CM-HNP.

In this model one can think of the bit security problem for any secret element $s$ (not necessarily a Diffie--Hellman key as the interaction with the oracle does not come from the key exchange setting). Let $\textbf{s} = (s_1,\ldots,s_n) \in G$ and write $\phi(\textbf{s})=(\phi_1(\textbf{s}),\ldots,\phi_n(\textbf{s})) \in G'$, and suppose that the oracle $O$ returns a bit of some component $\phi_i(\textbf{s})$. Write also $\textbf{r}=(r_1,\ldots,r_n)$.
Suppose there exists a family of homomorphisms $\phi^\textbf{r}$ for every\footnote{It is sufficient that there is a `large enough' subfamily of homomorphisms.} $\textbf{r}$ such that for some $1 \leq i \leq n$ the $i$-th component of $\phi^\textbf{r}$ satisfies $\phi^\textbf{r}_i(\textbf{x}) := \sum_{j=1}^n r_j x_j$. Then, getting a single bit of $\phi^\textbf{r}_i(\textbf{s}) = \sum_{j=1}^n r_j s_j$ for chosen $\textbf{r}$, gives rise to CM-MVHNP for a single-bit function and the secret $\textbf{s}$. A special case is where $\textbf{r}$ is of the form $r_j \cdot \textbf{e}\boldsymbol{_j}=(0,\ldots,0,r_j,0,\ldots,0)$. Then $\phi^\textbf{r}_i(\textbf{s}) = r_js_j$, which gives rise to CM-HNP for a single-bit function and secret $s_j$.

Therefore, if one can find a group for which the condition on the homomorphisms $\phi^\textbf{r}$ holds, then proving the hardness of single bits in this model reduces to either CM-MVHNP or CM-HNP (note that in the latter case one only recovers a component of $s$, and therefore needs other methods for recovering the entire value $s$; for the case in which $s = g^{ab}$ is a Diffie--Hellman key in $\F_{p^m}$ that we describe below, one can use the results involving ``summing functions'' from \cite{Eric-Verheul} and recover the entire secret $\textbf{s}$ from the algorithm that recovers a single (fixed) component $s_i$; for the case of elliptic curves it is sufficient to know one coordinate, as there are at most $3$ values for the other coordinate). We give a brief overview of the known results in the literature.

As mentioned above, this idea was introduced by Boneh and Shparlinski~\cite{BS} for the LSB of (both the $x,y$ coordinates of) Diffie--Hellman keys in elliptic curve groups over prime fields.
% (and recently for hyperelliptic curve~\cite{HyperEDH}).
It is shown there that changing the Weierstrass equation is an isomorphism that gives rise to the desired multipliers. Indeed, it is well known that twists of the curve give $\phi(x,y) = (u^2x,u^3y)$. Therefore, given a request for desired multiplier $r$ (for example by the SFT algorithm), one can obtain it if there is a solution to $u^d = r$ (where $d=2$ or $d=3$, depends on the coordinate) and flip a coin to guess the bit if a solution does not exist. The work~\cite{BS} uses the same technique as in \cite{ACGS} to prove hardness of LSB. This approach was then applied by \cite{Duc} (see also \cite{Kiltz}) to every single bit of a larger class of elliptic curve secrets, that also includes Diffie--Hellman keys in elliptic curves, using the SFT algorithm (that is, using the solution to CM-HNP for single-bit functions, as in Theorem~\ref{HNP_thm}).

The idea of changing group representations can also be used for finite fields. The works \cite{Fazio,WZZ} consider the computational Diffie--Hellman (CDH) problem in groups $\F_{p^m}^*$ for $m>1$. They show that some polynomial representations of $\F_{p^m}$ give rise to the desired homomorphisms $\phi^\textbf{r}$ for $\textbf{r} = r_j \cdot \textbf{e}\boldsymbol{_j}$, and therefore reduce to CM-HNP.

For a detailed overview of these techniques we refer the reader to the exposition of Sections $5, 5.1, 5.2$ and subsections within of \cite{MVHNP}. This latter work gives applications of the solution for CM-MVHNP to show bit security of the computational Diffie--Hellman problem in groups of higher dimension in models similar to those mentioned above; specifically, for elliptic curves over extension fields, and for $\F_{p^m}^*$ with different (non-polynomial) representations of the field $\F_{p^m}$.

We stress that these models do not tell a lot about the hardness of specific bits in real-life implementations of Diffie--Hellman key exchange, where the representation of the group is fixed. One should interpret results in the representation changing model as follows: assuming hardness of CDH in a group $G$ (where $G$ can be the multiplicative group of a finite extension field or an elliptic curve over a finite field), there is no algorithm that takes $g, g^a,g^b \in G$ and outputs the $i$-th bit of $g^{ab}$ for many representations of $G$ (more precisely, for representations corresponding to the specific isomorphisms used in the reduction). Nevertheless, given an instance $g^a,g^b$ in a specific representation of $G$, this result does not tell us whether it is hard to compute a specific bit of the secret $g^{ab}$. Indeed, this problem is still open.

\subsubsection{Sample-preserving search-to-decision reductions for LWE and LWR}

We assume the reader is familiar with the \emph{search} and \emph{decision} variants of the LWE and LWR problems~\cite{LWR,LWE}. The problem at hand is to reduce the search problem to the decision problem. That is, to show that the decision problem is at least as hard as the search problem. This is done in a similar fashion to the bit security reductions above: one assumes an oracle to the decision problem is available, and uses it to solve the search problem. We explain the reduction and show how the SFT algorithm is used to get a reduction in the stronger ``sample preserving'' model. This is done, as above, by reducing the problem to CM-MVHNP.

We only focus on the part of the reduction which involves the SFT algorithm; the entire reduction is more involved. By a ``hybrid'' argument (see \cite[Theorem 1]{HardCoreSurvey} or \cite[Lemma 3]{LWR_reduction}), one can reduce the decision problem to distinguishing a specific LWE sample, among the set of all samples.\footnote{The reduction given in \cite{MM} uses the duality of the LWE and knapsack functions.} We therefore consider a single LWE sample.

The standard method to show that the decision problem is as hard as the search problem is as follows. Suppose one has a perfect decision oracle. Given an LWE sample $b = \langle \textbf{a},\textbf{s} \rangle + e = a_1s_1 + \ldots + a_ns_n + e \pmod{p}$ one makes a guess $s'$ for $s_1$ and re-randomises the sample as $\textbf{a}' = ( a_1 + r, a_2, \dots, a_n), b' = b + rs' \pmod{p}$. If the guess is correct (i.e., if $s' = s_1$) then $( \textbf{a}', b')$ is a valid LWE sample whereas if the guess is incorrect then $b'$ is uniform and independent of the other smaples.
Hence the decision oracle determines whether the guess $s'$ of the secret value $s_1$ is correct.
After at most $pn$ queries to the decision oracle one can compute the secret $\textbf{s}$.

When the oracle is not perfect one will have to repeat this procedure with different inputs $(\textbf{a},b)$ and follow majority rule. When the success rate of the oracle is low, one may not have enough initial inputs $(\textbf{a},b)$ to satisfactorily apply the majority rule, and therefore would need to draw more samples. A \emph{sample-preserving} reduction is a reduction that uses only the initial given samples, and does not ask for more samples during the procedure. Micciancio and Mol~\cite{MM} used the SFT algorithm to give a sample-preserving search-to-decision reduction for the learning with errors problem. We now explain this reduction.

The standard method above involves choosing a unit vector $\textbf{e}_{j}$ and guessing $\langle \textbf{e}_{j},\textbf{s} \rangle$.
Micciancio and Mol observe that one can choose any vector $\textbf{v}$ and guess $\langle \textbf{v},\textbf{s} \rangle$, then let the decision oracle to advise whether this guess is correct or incorrect.
%Notice that one can try all possible $p$ guesses for the same value $\langle \textbf{v},\textbf{s} \rangle$, and store the one on which the oracle replied that the guess is correct, or keep drawing new vectors $\textbf{v}$ and make only one guess for $\langle \textbf{v},\textbf{s} \rangle$, denoted by $b_\textbf{v}$. The latter approach is taken in \cite{MM}, where
Again, if the oracle is perfect then one determines the correct guesses, denotes them by $b_\textbf{v}$, and eventually obtains $n$ linear equations in $\textbf{s}$ and hence can solve the problem. However if the oracle is not perfect (but has a noticeable advantage over a random guess), then in the case where the oracle says that the guess for $\langle \textbf{v},\textbf{s} \rangle$ is incorrect (more precisely, that the distribution is uniform), one sets $b_\textbf{v}$ to be some value from the remaining $p-1$ possibilities, chosen uniformly. Then for a selection of chosen vectors $\textbf{v}$ we have the values $b_\textbf{v}$, for which $b_\textbf{v} = \langle \textbf{v},\textbf{s} \rangle$ with some noticeable bias from $\frac{1}{p}$. In other words, we have query access to a noisy version of the function $f( \textbf{v} ) = \langle \textbf{v},\textbf{s} \rangle \pmod{p}$.

This is an instance of CM-MVHNP with an unreliable oracle. The function $\omega_p^{b_\textbf{v}}$, which is a noisy version of $\omega_p^{\langle \textbf{v},\textbf{s} \rangle}$, has a significant coefficient for the character $\chi_\textbf{s}$ (see Section~\ref{sec:e-concentrated}). Thus, one can run the SFT algorithm on the function $\omega_p^{b_\textbf{v}}$, to find this significant coefficient, hence the character, and thus solve this problem.
%Running the SFT algorithm in this setting boils down to guessing $\langle \textbf{v},\textbf{s} \rangle$ and $s_1$ and therefore calling the decision oracle on the values $\langle \textbf{a}+r\textbf{v},\textbf{s} \rangle$ and $\langle \textbf{a}+r\textbf{v}+ r'\textbf{u}_{1},\textbf{s} \rangle$ for uniformly chosen $r, r' \in \Z_p$.
%, where the latter corresponds to a guess $s'_1$ of $s_1$. Notice that there are $p$ guesses for the first inner product, for which one can show that only the correct one gives a valid LWE sample.\footnote{As in the previous footnote, one actually guess $\langle \textbf{v},\textbf{s} \rangle$ and then asks the decision oracle on $\textbf{a+}r\textbf{v}$ for a random number $r$.} Then there are $p$ possible guesses $s'_1$ as above.

A very similar approach is taken in \cite{LWR_reduction} for the learning with rounding problem. We remark that in the case of a non-prime $p$ the reduction is more subtle, and requires some restrictions (see~\cite{MM,LWR_reduction} for more details). We also remark that the reduction is an average-case reduction, and does not hold for the worst case (more precisely, there may be a set of initial samples $\{(\textbf{a}_i,b_i)\}$ for which the reduction fails). A sample-preserving reduction for the latter is still an open problem.

%Note that the scaling property is not used in this application. Instead, the method is directly learning the secret from a noisy version of the function $\omega_p^{\langle \textbf{v},\textbf{s} \rangle}$.
%More formally, this sketch shows that one can choose any vector $\textbf{r}$, and if the decision oracle has a non-negligible advantage over a trivial guess, then she gets the correct value of $\langle \textbf{r},\textbf{s} \rangle$ with non-negligible advantage over a trivial guess. This gives rise to CM-MVHNP, though \cite{MM} did not use the result on CM-MVHNP, as the corresponding `scaling property' is not needed here: one defines the (concentrated) function $f(\textbf{x}) := \chi_{\textbf{s}}(\textbf{x}) =  \omega_p^{\langle \textbf{x},\textbf{s} \rangle}$, for which $\textbf{s}$ is a heavy coefficient, and therefore only needs to find the heavy coefficients of this function. For more details see \cite{MM}.

\section{Limitations: Non-Linear Problems}
\label{sec_limitations}

This section presents limitations on natural generalisation of the approaches taken above to a larger class of applications. We show that the linearity in the hidden number problem, induced from the operation $s \cdot x$, is essential for the SFT to be useful. In particular, we give an answer (in the negative) to an open question in~\cite{MIHNP}.

The solution to the CM-HNP in $\Z_N$ (Theorem \ref{HNP_thm}) is based on Fourier analysis in the additive group $(\Z_N,+)$ and it exploits the scaling property of the Fourier transform for the function $f_s(x) := f(sx)$. In other words, the function $f_s$ is the composition of $f$ with a linear map on $\Z_N$.
It is natural to consider whether this approach can be used for other algebraic groups (such as elliptic curves and algebraic tori).
The hidden number problem in the case of elliptic curves is to determine a secret point $S \in E( \F_p )$ given samples $( P, f( S + P ))$ where a typical choice for the function would be $f(Q) = \bit_i( x(Q))$.
The natural approach is to still use Fourier analysis in the additive group $(\Z_p,+)$ but instead of composing with a linear map, to compose with a rational function (e.g., coming from the translation map $t_S(P) = P+S$). Another generalisation would be Fourier
analysis in other groups $(G, \cdot)$.

If such tools could be developed we might have an approach to the bit security of Diffie--Hellman key exchange in the group of elliptic curve points in certain models.
% that, unlike the methods discussed in Section~\ref{sec:applications}, would not require changing representations. (DO WE ASSUME TO CHOOSE TO `MULTIPLIERS' HERE? I THINK WE SHOULD SOFTEN THE CLAIM)
There are also other interesting problems that could be approached with Fourier analysis on general groups. For example, the authors of~\cite{MIHNP} raise the question whether it is possible to apply these results to the \emph{modular inversion hidden number problem}.

Unfortunately, there is a major obstacle to applying the SFT algorithm to these sorts of problems.
Namely, if $f$ is a concentrated function then the composition $f \circ \varphi$ is concentrated only when $\varphi$ is affine. In fact, $f \circ \varphi$ has significant coefficients only when $\varphi$ is affine.
The aim of this section is to explain this obstacle.
Since the translation map for the elliptic curve group law is a non-affine rational function, this explains why the method cannot be directly applied to the elliptic curve hidden number problem.
Our argument also answers the question of~\cite{MIHNP} in the negative.

\bigskip

Let $f : G \to \C$ be a function and let $f_s(x) = f \circ \varphi_s (x)$, where $\varphi_s : G \to G$ is an efficiently computable function (that depends on some unknown value $s$). To generalise the proof of Theorem~\ref{HNP_thm} one needs the following three conditions:

\begin{enumerate}
    \item the function $f$ has significant coefficients;
    \item the function $f_s$ has significant coefficients;
    \item there exists a relation between the significant coefficients of $f$ and $f_s$ that allows to determine $s$ (or at least a small set of candidates for $s$).
\end{enumerate}

%Note that Theorem \ref{HNP_thm} does not require a concentrated function, but only a significant coefficient (EPS-CONCENTRATED?); we address this difference below.

One special case is when $f$ is a constant function. Then $f_s$ is also a constant function and both conditions 1 and 2 are satisfied. The problem is that a constant function cannot tell us anything about the secret $s$, and so condition 3 does not hold.
Hence, we need to focus on functions that are far from constant, which we formalise in our proof by requiring that $\widehat{f}(0) = 0$ (in other words, $f$ is ``balanced'').

Having dispensed with this special case we focus on the first two conditions.
We first consider the case when $f$ is concentrated.
If $\varphi_s(x) = ax + b$ is affine then we already know from the scaling and time-shifting properties that all Fourier coefficients of $f$ are preserved in $f_s$, and so if $f$ is concentrated then $f_s$ is also concentrated.
Our aim is to show a converse to this fact: if $\varphi_s$ is a rational function and if conditions 1 and 2 both hold then $\varphi_s$ must be affine.
%In other words preserving concentration is a global property, that is, it is independent of the function $f$.
This result is closely related to the Beurling--Helson Theorem \cite{Beur-Hel} (see \cite{KonShk,Lebedev} for related results in $\Z_p$) and the work of Green and Konyagin \cite{GreKon} on the Fourier transform of balanced functions.

For our result we need the following lemma \cite[Lemma 7]{NguShpa} (a proof, for general fields $\F_{p^m}$, can be found in  \cite[Theorem 2]{Moreno}).

\begin{comment}
\begin{lemma} \label{lem_character_sum}
    Let $\varphi = \frac{g}{h}$ be a non-constant quotient of two univariate polynomials with coefficients in $\F_{p^m}$, let $Z_h$ denote the set of roots of $h$, and let $\chi$ be a non-trivial additive character of $\F_{p^m}$. The following bound holds
    \[ \left| \sum_{x \in \F_{p^m} \setminus Z_h} \chi\left(\varphi(x)\right) \right| \leq \left(\max\left\{\deg(g), \deg(h)\right\} + u - 2 \right) \sqrt{p^m} + \delta \enspace, \]
    where
    \begin{displaymath}
        (u,\delta) = \left\{
        \begin{array}{l l}
            (v,1)   & \quad if \ \deg(g) \leq \deg(h),\\      (v+1,0) & \quad if \ \deg(g)>\deg(h),
        \end{array}
        \right .
    \end{displaymath}
    and $v$ is the number of distinct roots of $h$ in the algebraic closure of $\F_{p^m}$.
\end{lemma}
\end{comment}

\begin{lemma}\label{lem_character_sum}
    Let $q$ be prime. For any polynomials $f,g \in \F_q[x]$ such that the rational function $h = \frac{f}{g}$ is not constant in $\F_q$, the following bound holds
    \[ \Bigg| \sum_{\lambda \in \F_p}{}^{*} \omega_q^{h(\lambda)} \Bigg| \leq (\max\{\deg(f), \deg(g)\} + u -2) \sqrt{q} + \delta \,, \]
    where $\sum^*$ means that the summation is taken over all $\lambda \in \F_q$ which are not poles of $h$ and
    \begin{displaymath}
        (u,\delta) = \left\{
        \begin{array}{l l}
            (v,1)   & \quad if \ \deg(f) \leq \deg(g),\\      (v+1,0) & \quad if \ \deg(f)>\deg(g),
        \end{array}
        \right .
    \end{displaymath}
    and $v$ is the number of distinct zeros of $g$ in the algebraic closure of $\F_q$.
\end{lemma}

We formulate the following result for functions on $\Z_q$ for a prime $q$, but it can be generalised to finite fields $\F_{p^m}$ with $m > 1$.
Let $g, h \in \Z_{q}[x]$ be polynomials where $h$ is not the constant zero.
Let $Z_h$ be the set of zeroes in $\Z_{q}$ of $h$. We define
$\varphi(x) = g(x)/h(x)$ for all $x \in \Z_q \setminus Z_h$ and $\varphi(x) = 0$ otherwise (since we will assume $Z_h$ is small compared with $q$ it does not matter how we define $\varphi$ on $Z_h$).

Recall that the definition of concentration applies to families of functions. To keep the formulation of the following proposition clean, we call a single function concentrated as explained after the definition above.

\begin{proposition}\label{prop_limit}
    Let $q$ be a sufficiently large prime.
    Let $f$ be a concentrated function on $\Z_q$ such that $\|f\|_2=1$ and $\widehat{f}(0) = 0$.
    Let $g, h \in \Z_q[x]$ be polynomials of degree bounded by $poly(\log(q))$ and let $Z_h$ be the set of zeroes of $h$.
    Define $\varphi(x)$ as above and suppose this function is non-constant.
    Let $\tau = 1/poly(\log(q))$. If $f \circ \varphi$ has any $\tau$-heavy Fourier coefficients then $\varphi(x) = ax + b$ for some $a, b \in \Z_q$.
\end{proposition}

\begin{proof}
    Let $G = \Z_q$ and write $f  = \sum_{\alpha \in G} \widehat{f}(\alpha) \chi_\alpha$.
    Let $d = \max\{ \deg(g(x)), \deg(h(x)) \}$.
    %For contradiction we suppose $\varphi(x) \ne ax + b$ for any $a, b$.
    Let $\epsilon = \frac{\tau}{32 d^2}$. Since $f$ is concentrated there is a set $\Gamma$ of size $poly(\log(\abs{G}))$  such that
    \[
        \| f - f|_\Gamma \|_{2}^{2} \leq \epsilon = \frac{\tau}{32 d^2} \,.
    \]
    Since $\widehat{f}(0) = 0$ it follows that $\Gamma$ does not contain zero.

    Now consider $f_\varphi(x) = f(\varphi(x)) = \sum_{\alpha \in G} \widehat{f}(\alpha) \chi_\alpha(\varphi(x))$.  Assume it has a $\tau$-heavy coefficient; for contradiction we suppose $\varphi(x) \ne ax + b$ for any $a, b$. For every $\beta \in G$ we have
    \begin{equation*}
        \begin{split}
            \widehat{f_\varphi}(\beta) = & \frac{1}{|G|} \sum_{x \in G} f_\varphi(x) \overline{\chi_\beta(x)} = \frac{1}{|G|} \sum_{x \in G} f(\varphi(x)) \overline{\chi_\beta(x)} = \\
            & \frac{1}{|G|} \sum_{x \in G} \sum_{\alpha \in G} \widehat{f}(\alpha) \chi_\alpha(\varphi(x)) \overline{\chi_\beta(x)} = \frac{1}{|G|} \sum_{\alpha \in G} \widehat{f}(\alpha) \sum_{x \in G} \chi_\alpha(\varphi(x)) \overline{\chi_\beta(x)} = \\
            & \frac{1}{|G|} \sum_{\alpha \in G} \widehat{f}(\alpha) \sum_{x \in G} \chi_1(\alpha \varphi(x) - \beta x) = \frac{1}{|G|} \sum_{\alpha \in G} \widehat{f}(\alpha) \sum_{x \in G} \chi_1(\psi^\beta_\alpha(x)) \, ,
        \end{split}
    \end{equation*}
    where we denote $\psi^\beta_\alpha(x) = \alpha \varphi(x) - \beta x$.
    Since $\widehat{f}(0) = 0$ we can ignore the case $\alpha = 0$ and by our supposition that $\varphi \ne ax+b$ we know that there are no $\alpha, \beta$ such that $\psi^\beta_\alpha$ is constant.
    Hence, the last sum is a character sum satisfying the conditions of Lemma~\ref{lem_character_sum}.
    Furthermore, $\psi^\beta_\alpha = (\alpha g(x) - \beta x h(x))/h(x)$ and so the value $u$ in Lemma~\ref{lem_character_sum} is bounded by $\max\left\{\deg(g), \deg(h)\right\} \leq d$.
    %Hence the character sums are all bounded by $C = 2 \log(q) \sqrt{q}$.
    Applying Lemma \ref{lem_character_sum}, we get that for every $\alpha \ne 0$ and every $\beta$ it holds that $|\sum_{x \in G \setminus Z_h} \chi(\psi^\beta_\alpha(x))| \leq C$ where $C = 2 d \sqrt{q}$.

    Now note that
    \[ \widehat{f_\varphi}(\beta) = \frac{1}{|G|} \sum_{\alpha \in G} \widehat{f}(\alpha) \sum_{x \in Z_h} \chi_1(\psi^\beta_\alpha(x)) + \frac{1}{|G|} \sum_{\alpha \in \Gamma} \widehat{f}(\alpha) \sum_{x \in G \setminus Z_h} \chi_1(\psi^\beta_\alpha(x)) + \frac{1}{|G|} \sum_{\alpha \notin \Gamma} \widehat{f}(\alpha) \sum_{x \in G \setminus Z_h} \chi_1(\psi^\beta_\alpha(x)) \, .
    \]
    For the first term we note that $|\sum_{x \in Z_h} \chi_1(\psi^\beta_\alpha(x))| \le | Z_h| \le d$ and that $\Vert f \Vert_2 = 1$ implies $\sum_{\alpha \in G} |\widehat{f}(\alpha)| \le \sqrt{|G|} = \sqrt{q}$ and $|\widehat{f}(\alpha)| \leq 1$ for all $\alpha$.
    Therefore
    \[ \left| \widehat{f_\varphi}(\beta) \right| \leq \frac{d}{\sqrt{q}} + \left| \frac{1}{|G|} \sum_{\alpha \in \Gamma} \widehat{f}(\alpha) \sum_{x \in G \setminus Z_h} \chi(\psi^\beta_\alpha(x)) \right| + \left| \frac{1}{|G|} \sum_{\alpha \notin \Gamma} \widehat{f}(\alpha) \sum_{x \in G \setminus Z_h} \chi(\psi^\beta_\alpha(x)) \right| . \]

    We apply the triangle inequality on the first sum and the Cauchy--Schwarz inequality on the second. Let $k = |\Gamma|$ and write $\Gamma = \{ \alpha_1, \dots, \alpha_k \}$.
    Then using Lemma \ref{lem_character_sum} we get
    \begin{equation*}
        \begin{split}
            \left| \frac{1}{|G|} \sum_{\alpha \in \Gamma} \widehat{f}(\alpha) \sum_{x \in G \setminus Z_h} \chi(\psi^\beta_\alpha(x)) \right| & = \left| \frac{1}{|G|} \sum_{j=1}^{k} \widehat{f}(\alpha_j) \sum_{x \in G \setminus Z_h} \chi(\psi^\beta_{\alpha_j}(x)) \right| \leq \left| \frac{1}{q} \sum_{j=1}^{k} \widehat{f}(\alpha_j) \cdot C \right| \\
            & \leq \frac{1}{q} \sum_{j=1}^{k} \left| \widehat{f}(\alpha_j) \right| C = \frac{2k d}{\sqrt{q}} \, .
        \end{split}
    \end{equation*}
    %for $C_j = O(1)$ (XXX Previously written incorrectly as $C_j << \sqrt{q}$ XXX). Let $C = \max\{C_j\}_{1 \leq j \leq k}$, then $\left| \frac{1}{|G|} \sum_{\alpha \in \Gamma} \widehat{f}(\alpha) \sum_{x \in G \setminus Z_h} \chi(\psi^\beta_\alpha(x)) \right| \leq \frac{kC}{\sqrt{q}}$.
    Since $k = | \Gamma | = poly( \log(q))$ we have that this bound (similarly for the earlier bound $d/\sqrt{q}$) is negligible, so we have for example
    % (XX FORMULA QUITE ARBITRARY. Just $\sqrt{\epsilon}$ would suffice XX)
    \[
        \frac{d}{\sqrt{q}} + \frac{2kd}{\sqrt{q}} < 2 d \sqrt{\epsilon} \,.
    \]

    From Parseval's identity $\sum_{\alpha \notin \Gamma} \left| \widehat{f}(\alpha) \right|^2 = \|f - f|_\Gamma \|_{2}^{2} \leq \epsilon$. Therefore, by the Cauchy--Schwarz inequality we have
    \[
        \left| \frac{1}{|G|} \sum_{\alpha \notin \Gamma} \widehat{f}(\alpha) \sum_{x \in G \setminus Z_h} \chi(\psi^\beta_\alpha(x)) \right|
        \leq
        \frac{1}{|G|} \left( \sum_{\alpha \notin \Gamma} \left| \widehat{f}(\alpha) \right|^2 \right)^\frac{1}{2} \left( \sum_{\alpha \notin \Gamma} \left| \sum_{x \in G \setminus Z_h} \chi(\psi^\beta_\alpha(x)) \right|^2 \right)^\frac{1}{2} \leq \frac{1}{|G|} \sqrt{\epsilon} \left(\sum_{\alpha \notin \Gamma} C^2  \right)^\frac{1}{2} .
    \]
    Then
    \[\left| \frac{1}{|G|} \sum_{\alpha \notin \Gamma} \widehat{f}(\alpha) \sum_{x \in G \setminus Z_h} \chi(\psi^\beta_\alpha(x)) \right| \leq \frac{\sqrt{\epsilon} \sqrt{q-k} 2 d \sqrt{q}}{q} \leq 2 d \sqrt{\epsilon} \,. \]

    Finally, combining the bounds we get
    \[ \left| \widehat{f_\varphi}(\beta) \right|^2 \leq \left(\frac{d}{\sqrt{q}} + \frac{2kd}{\sqrt{q}} + 2 d \sqrt{\epsilon}\right)^2 < \left(4 d \sqrt{\epsilon}\right)^2 = \left(4d \frac{\sqrt{\tau}}{4 d \sqrt{2}}\right)^2 = \frac{\tau}{2} \, . \]
    Therefore, for every $\beta$ the coefficient $\widehat{f_\varphi}(\beta)$ is not $\tau$-heavy for any noticeable $\tau$. %\footnote{Alternatively, one can take $\tau^{-1} = subexp(\log q)$, and so $k= |\Gamma| = subexp(\log q)$, to show that $\left| \widehat{f_\varphi}(\beta) \right|^2 < \dfrac{subexp(\log q)}{\sqrt{q}} + \dfrac{1}{2subexp(\log q)} < \dfrac{1}{subexp(\log q)}$??.}
    This gives the required contradiction and so we conclude that
    $\varphi$ is affine.
\end{proof}

\subsection{$\epsilon$-concentrated functions}\label{sec:e-concentrated}

Proposition~\ref{prop_limit} shows that if $f$ is concentrated (and far from constant) and $f \circ \varphi$ has significant coefficients, then $\varphi$ is affine.
It is natural to wonder whether the condition that $f$ is concentrated is necessary.
In fact, the result cannot be weakened in general: if $\varphi(x) = g(x)/h(x)$ is non-affine and invertible almost everywhere (such as a M{\" o}bius function $\varphi(x) = (ax + b)/(cx + d)$ where $ad - bc = 1$) then $f(x) = \chi_\alpha(x) + \chi_\beta(\varphi^{-1}(x))$ is such that $f(x)$ has a significant coefficient at $\alpha$ and $f \circ \varphi$ has a significant coefficient at $\beta$.

However, a version of Proposition~\ref{prop_limit} is true for some non-concentrated functions of interest. Since Theorem~\ref{HNP_thm} does not require the function to be concentrated, it is of interest to also show that composing with non-affine $\varphi(x)$ is an obstruction to the solution to CM-HNP for these functions as well.
Hence, for the rest of this section we consider a `noisy character', $f(x) := \omega_N^{\alpha x + e(x)}$.
We first show that these functions have a significant coefficient, then we show that $f \circ \varphi$ does not have a significant coefficient when $\varphi$ is not affine.
%, all functions of interest are concentrated. It is hard to imagine a  naturally-occurring function for which the claim in Proposition \ref{prop_limit} does not hold. We now prove the claim in Proposition \ref{prop_limit} for noisy characters.
%to some other family of functions.

To formalise the problem we think of $e(x)$ as a random variable from some distribution (e.g., a discrete Gaussian or a uniform distribution on some small interval compared with $N$).
%
%Consider a function given by $f_\varphi(x) := f \circ \varphi(x) =  \omega_N^{\varphi(x) + e(x)}$ where the $e(x)$ is an error term given by some distribution. The distribution of $e$ clearly affects the analysis of the Fourier coefficients of $f$. From the applications point of view, the most interesting distributions are the uniform distribution over an interval in $Z_N$, which appears in bit-security-related hidden number problems and LWR, and the Gaussian distribution, which appears in LWE variants. We restrict this discussion to these distribution.
We treat $e(x)$ as being independent of $x$, in which case we can write
\[
    \widehat{f}(\beta)=\E\left(\omega_N^{\alpha x -\beta x + e(x)}\right)=\E\left(\omega_N^{( \alpha - \beta) x }\omega_N^{e(x)}\right)= \E\left(\omega_N^{( \alpha - \beta)x}\right)\E\left(\omega_N^{e(x)}\right) .
\]
To show that $|\widehat{f}(\alpha)|$ is large it suffices to give a lower bound for $ \big| \E\big(\omega_N^{e(x)}\big) \big|$.
We do this by following an argument due to Bleichenbacher~\cite{Bleichenbacher}.

Bleichenbacher defines the \emph{bias} of a random variable $X$ on $\Z$ as
\[
    B_N(X) = \E\left( \exp( 2\pi i X / N ) \right) .
\]
Assume $X$ is the uniform distribution in some interval $[0,T-1]$ for some $0<T\leq N$.
Then
\[ B^U_N(X) := B_N(X)  = \frac{1}{T}\sum_{0\leq x<T}\exp( 2\pi i x / N ) \,. \]
Some properties of $B^U_N(X)$ appear in Lemma $1$ of \cite{De-Mulder}. Since the latter is a geometrical progression,
\[ B^U_N(X) = \frac{1}{T}\frac{\sin(\pi T /N)}{\sin(\pi N)} \,. \]

Suppose $e(x)$ follows the uniform distribution $X$. That is, for each $x \in \Z_N$ the value $e(x)$ is chosen uniformly and independently at random in $[0,T-1]$. From linearity it is easy to see that
\[ \E\left(\omega_N^{e(x)}\right) = \frac{1}{N}\sum_{x\in\Z_N}\left(\exp( 2\pi i e(x) / N)\right) = \frac{N/T}{N}\sum_{0\leq t<T}\exp( 2\pi i t / N )  = \frac{1}{T}\sum_{0\leq t<T}\exp( 2\pi i t / N ) =B^U_N(X) \,. \]

It is obvious that if $T=N$ then $B^U_N(X)=0$. In applications $e(x)$ usually represents some given bits, and so it is natural to restrict $T\leq N/2$ as we do, though the following argument also holds given a ``fraction of a bit'', i.e. for $T>N/2$. For $T\leq N/2$ one has\footnote{See \cite[Table 1]{De-Mulder} for some values $|B^U_N(X)|$ for different $T\leq N/2$.} $|B^U_N(X)|>0.5$, and so $\big|\E\big(\omega_N^{e(x)}\big)\big|^2=|B^U_N(X)|^2>0.25$. The desired lower bound is provided.

%To show that the noisy character $f(x) = \omega_N^{\alpha x + e(x)}$ has a significant coefficient, consider $\widehat{f}(\beta) = \E\left(\omega_N^{(\alpha-\beta) x + e(x)}\right)=B_N((\alpha-\beta) x + e(x))$ and notice that similarly to \cite[Lemma 1(a)]{De-Mulder}, as this is a sum of independent variables we have $\widehat{f}(\beta)=B_N((\alpha-\beta) x )B_N(e(x)) = \widehat{\omega_N^\alpha}(\beta)B_N(e(x))$.

%Similarly, when one considers $f_\varphi = \omega_N^{\alpha \varphi(x) + e(x)}$ we get $\widehat{f_\varphi}(\beta)= \widehat{\omega_N^\alpha\varphi(x)}(\beta)B_N(e(x))$. When $e(x)$ is uniform in $[0,T-1]$ then
%\[|\widehat{f}(\beta)|^2=|\widehat{\omega_N^\alpha}(\beta)B^U_N(e(x))|^2 \leq \frac{1}{4}|\widehat{\omega_N^\alpha}(\beta)|^2 \] and
%\[ |\widehat{f_\varphi}(\beta)|^2=|\widehat{\omega_N^\alpha\varphi(x)}(\beta)B^U_N(e(x))|^2 \leq \frac{1}{8}\tau \]
%for any noticeable $\tau$ as shown above.

A similar approach holds when $e$ follows a Gaussian distribution. In this case the size of the bias is even larger, as $e(x)=0$ on a large set (and $e(x)$ is small on an even larger set) and so most of the ``energy'' is distributed around zero.

Hence, we have established that a noisy character has a significant coefficient.
Finally, we address the result of Proposition~\ref{prop_limit} for such a function.

\begin{claim}
    Let $\varphi$ be as in Proposition \ref{prop_limit}, and let $e(x)$ given by the uniform distribution (over some interval in $\Z_N$) or by a Gaussian distribution. If $f_{\varphi}(x):= \omega_N^{\varphi(x) + e(x)}$ has a significant coefficient then $\varphi(x) = ax + b$ for some $a, b \in \Z_N$.
\end{claim}

\begin{proof}[Proof sketch]
    We observe that for every $\beta$
    \[ \widehat{f_{\varphi}}(\beta)=\E\left(\omega_N^{\varphi(x)-\beta x + e(x)}\right)=\E\left(\omega_N^{\psi^\beta_1(x)}\omega_N^{e(x)}\right)= \E\left(\omega_N^{\psi^\beta_1(x)}\right)\E\left(\omega_N^{e(x)}\right) , \]
    where $\psi^\beta_1(x) = \varphi(x) - \beta x$.
    Since $\Big|\E\left(\omega_N^{e(x)}\right)\Big| \leq 1$, it suffices to upper-bound
    $\Big|\E\left(\omega_N^{\psi^\beta_1(x)}\right)\Big| $.
    Such a bound follows from Lemma \ref{lem_character_sum} in the same way as in the proof of Proposition~\ref{prop_limit}.
\end{proof}

%Moreover, we show that when $e$ is given by the uniform distribution (over an interval in $Z_N$) or by the Gaussian distribution, the function $f(x) = \omega_N^{\alpha x + e(x)}$ is $\epsilon$-concentrated with $|\widehat{f}(\alpha)|^2 \geq \epsilon$. In this case $\E\big(\omega_N^{\psi^\alpha_1(x)}\big)=1$, so one only needs to give a lower bound on $\big|\E\big(\omega_N^{e(x)}\big)\big|$.

\subsection{Hidden number problem in subgroups}
\label{sec_subgroups}

Another limitation on the applications of the SFT algorithm is the following. Suppose that the multipliers in the hidden number problem are drawn from some set $H \subseteq G$. One can consider the multipliers to be in a proper subgroup $H < G$, as done in \cite{GVS, SW}. It is not clear how to apply the SFT algorithm to solve this variant of the (chosen-multiplier) hidden number problem. Specifically, the chosen queries in the algorithm have to be correlated, but it is not guaranteed that these correlated queries will all lie in the same subgroup.
If the index $[G:H]$ is small (e.g., $[G:H] = 2$, as in the case of the set of squares in $\F_p^*$) then the issue can be managed, but if $[G:H]$ is large then no results are known.
%This applies to relatively small subgroups of $G$.
Therefore, for results (on Diffie--Hellman related schemes) that rely on advice of the form of discrete logarithms to some base $g$ (as in \cite{AKA-HNP,BV2,SW2}), if $g$ generates a relatively small subgroup, it is not guaranteed that the desired correlated multipliers are indeed in the group generated by $g$. This restricts, for example, the result given in \cite[Section 5]{AKA-HNP}. This observation is similar to the one in \cite[Section 2.5]{Shparlinski}, and was handled in \cite[Section 5]{BS} and \cite[4.1]{Duc} since the set of squares in $\F^*_p$ has index 2 in $\F^*_p$.

\section*{Acknowledgements}

We thank Ben Green for providing some insights and references.
We also thank two anonymous referees for their helpful comments on an earlier version of the paper.
% We also thank the helpful comments of two anonymous referees.

\bibliographystyle{tocplain}   %%% \bibliographystyle{plain}
%%% !!!
%%% Change this to match the name of your BIB file
%\bibliography{template}

\end{document}